\documentclass[12pt]{article}

\usepackage{cite}
\usepackage{amsfonts}
\usepackage{graphicx}
\usepackage{amsmath}
\usepackage{subfigure}

\newenvironment{definition}[1][Definition]{\begin{trivlist}
\item[\hskip \labelsep {\bfseries #1}]}{\end{trivlist}}
\newtheorem{theorem}{Theorem}[section]

\newenvironment{proof}[1][Proof]{\begin{trivlist}
\item[\hskip \labelsep {\bfseries #1}]}{\end{trivlist}}
\newcommand{\qed}{\nobreak \ifvmode \relax \else
      \ifdim\lastskip<1.5em \hskip-\lastskip
      \hskip1.5em plus0em minus0.5em \fi \nobreak
      \vrule height0.75em width0.5em depth0.25em\fi}

\date{}

\begin{document}



\title{Topological Fidelity in Sensor Networks}
\author{Harish Chintakunta and Hamid Krim}
\maketitle

\begin{abstract}
Sensor Networks are inherently \emph{complex networks}, and many of their associated problems require analysis of some of their global characteristics.  These are primarily affected by the topology of the network. We present in this paper, a general framework for a topological analysis of a network, and develop distributed algorithms in a generalized combinatorial setting in order to solve two seemingly unrelated problems, 1) Coverage hole detection and Localization and 2) Worm hole attack detection and Localization. We also note these solutions remain coordinate free as  no priori localization information of the nodes is assumed. For the coverage hole problem, we follow a ``divide and conquer approach'', by strategically dissecting the network so that the overall topology is preserved, while efficiently pursuing the detection and localization of failures. The detection of holes, is enabled by first attributing a combinatorial object called a "Rips Complex"  to each network segment, and by subsequently checking the existence/non-existence of holes by way of triviality of the first homology class of this complex. Our estimate exponentially approaches the location of potential holes with each iteration, yielding a very fast convergence coupled with optimal usage of valuable resources such as power and memory. We then show a simple extension of the above problem to address a well known problem in networks, namely the localization of a worm hole attack.  We demonstrate the effectiveness of the presented algorithm with several substantiating examples.
\end{abstract}






\section{\large{\textbf{Introduction}}}
\label{sec:Introduction}
The infrastructure of computing systems is rapidly transitioning from centralized systems to distributed and pervasive systems. A very important class of such systems are sensor networks which find applications in areas including Environmental monitoring, Health care  and Military operations \cite{WSNsurvey}. There has been a considerable research interest in this field over the past decade, addressing problems including node localization \cite{Localization}, distributed compression \cite{distCompression}, probabilistic inference \cite{GBP} and motion tracking. A unifying theme of many of these problems is to glean  consensus information by systematically combining the data collected at individual nodes, in accordance to the structure of the network. The consensus information thus obtained characterizes the network, or the data in the network as a whole, and better represents the underlying phenomenon which can be inferred from the data at individual nodes. This reveals the fundamental nature of sensor networks: they are essentially \emph{complex networks} in which global patterns emerge from simple interactions between nodes. From an engineering perspective, the fundamental challenge in sensor network applications is to cope with the limited resources; a limited communication capability of nodes, i.e. nodes can only communicate with their neighbors, with a limited power and a limited memory. Furthermore, sensor networks are often deployed in unaccessible locations and environments where maintenance is impractical; this makes careful use of exhaustible resources such as power, imperative.\\
This unique set of circumstances motivates the use of techniques such as topological analysis. This is to directly extract global information without being overly dependent on the local structure, and thereby alleviating the excessive need for recourses. We demonstrate in this paper, the merits of such analysis by exploiting tools to solve two specific important problems: 1)A Coverage Hole detection and localization and 2)A Worm-Hole Attack detection and localization. \\
The first Problem discussed in Section \ref{sec:CoverageProblem} seeks to identify an area within a network which is not in the range (and hence uncovered) of any sensor. The second problem investigates the detection and localization of an attack called a worm-hole. This has a potential of substantially disrupting routing, localization and other tasks in a network. In the next section, we endeavor to briefly summarize the research in topological analysis and work related to the techniques presented in this paper.

\subsection{\textbf{Topological Analysis in Sensor Networks}}
Distributed algorithms for analyzing topological properties may be broadly classified into three categories: Geometric, Topological, and Statistical Methods. This categorization is based on the taxonomy presented in \cite{WangGao}, and a good overview of algorithms in these areas is presented in \cite{FayedMouftahAlpha}. \\
Topology may be described as the study of arrangement of spaces (manifolds or other data spaces) whereas Geometry may be described as the study of metrics (measures of distance) on these spaces.  This distinction characterizes the difference between distributed algorithms using topological and geometric methods. In sensor networks, the space of interest is first constructed using the node parameters followed by an analysis methodology of choice (For example, in the coverage problem in Section \ref{sec:CoverageProblem}, the space of interest is the total coverage area). We may also view the Geometric methods as a ``fine''analysis of spaces whereas Topological methods as a ``coarse'' analysis. Statistical methods rely on the aggregate statistical behavior of node parameters and try to infer the necessary information of the network by tracing the changes in these statistics. Our present work falls into the category of a Topological approach.  The distinction of these methods can be instantiated by looking at existing algorithms to solve the coverage problem. \\
In geometric methods, for example, the work in \cite{FayedMouftahAlpha} computes an $\alpha$-hull of the node positions in order to identify  the outer and inner (coverage-hole) boundaries of a network. An $\alpha$-hull of a set of points $V$ in a plane is given by the intersection of complementary regions of circles of radius $1 / \alpha$, such that no point in $V$ lies inside these circles. The complement of a circle is defined as the entire plane excluding the interior of this circle. Some other geometric methods for the coverage problem are presented in \cite{FangGao} and \cite{FayedMouftahAlpha}. \\
An example of a statistical approach for a coverage problem may be found in \cite{FeketeKroller}. It relies on the idea that nodes close to network boundaries, have fewer incident edges in the network graph than internal nodes. The authors use statistical methods to derive suitable thresholds to separate edge nodes from internal nodes using the node degrees. In \cite{FeketeKaufmann}, boundary nodes are separated from internal nodes by using a centrality measure which counts the number of shortest paths that pass through a node. A higher centrality value occurs among internal nodes. \\
In the Topological methodology, Morse theory and Algebraic Topology are the most commonly used tools.

\subsubsection{\textbf{Morse Theoretic Methods}}
\label{subsubsec:MorseTheoreticMethods}
Morse theoretic methods analyze the topology of a given topological space, more specifically a manifold, by studying differentiable functions defined on it. Consider a differentiable function $f:{M}\rightarrow \mathbb{R}$ defined on a manifold ${M}$, then the inverse image of a point in $\mathbb{R}$  is called a level set.  A principal tenet underlying these methods is the observation that critical points of this function, where topology of the level sets changes, directly reflect the underlying topological construction of the space. \\
For the coverage problem, an example of a Morse theoretic topological method is given in \cite{Funke}. The authors find boundaries of a network by studying the behavior of connected components, the nodes of which are at equi-hop distance from a randomly selected point in the network. The main observation here is that each of these components has a discontinuity at the boundaries of the network.
%

Morse theoretic methods often provide simple and efficient (in complexity) methods to analyze the underlying topology, but the greatest challenge of these methods is often the construction of an appropriate function. In addition, since the theory is mostly developed for manifolds, it requires stricter assumptions for discrete spaces such as positions of nodes in a network. For example, the work presented in \cite{Funke} will fail to provide reasonable results if the node density is small, or their distribution is non-uniform.

\subsubsection{\textbf{An Algebraic Topological Approach}}
\label{subsubsec:AlgebraicTopologicalMethods}
Algebraic topology, in contrast to a Morse theory, is  a relatively more direct technique to analyze the topology of a space which is easily expressed in terms of algebraic objects. There is an extensive literature in Algebraic topology \cite{vick,hatcher} which shows a very strong  relationship between topological spaces and their algebraic counterparts. This also enables us to draw from an extensive source of knowledge in algebra to develop fast and efficient algorithms.
The algebraic objects of choice have the following important properties: They directly reflect the topological features of an underlying space, and they are invariant to continuous deformations. We therefore follow this approach in our work owing to these advantages.

\begin{figure}
\centering
\includegraphics[width=0.5\textwidth]{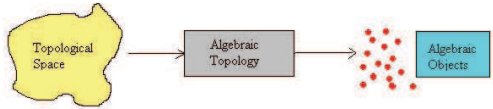}
\caption{A high level schematic of Algebraic Topology}
\label{fig:schematicAlgebraicTopology}
\end{figure}

The use of algebraic topology for the coverage problem mentioned above, was first introduced in \cite{deSilvaGhrist,RghristHomo}. The works in \cite{controlLap,distHomo} propose a distributed computation of homology groups, and \cite{DistLocHoles} attempts to localize the holes by formulating localization as an optimization problem. We further exploit the natural spatial constraint of the coverage holes, to formulate a new effective and efficient "divide and conquer" algorithm. To the best of our knowledge, \cite{DistLocHoles} is the first attempt at distributively localizing holes using an algebraic topological approach  and we compare our work on coverage holes with the results presented in there. A preliminary version of this paper was presented in \cite{SECONholeConvergence}. \\
A hole may be caused by  a) deployment error, b) some catastrophic event such as an explosion, c) presence of a jammer which disables associated nodes' communication, and thereby hiding their presence within the network. In case of deployment error, a hole localization helps in targeted redeployment, and more generally,  helps in precisely identifying the location of the cause of failure events. Routing algorithms such as geographic-routing \cite{geographicRouting} and some other distributed signal processing algorithms heavily depend on certain assumptions about the topology of the network of interest \cite{shapeSegmentation}. Having this knowledge of the overall topology may therefore be very useful.\\
Other network failures, which may have a devastating impact, include worm holes. A worm hole attack is typically launched by two colluding external attackers who do not authenticate themselves as legitimate nodes to the network. When initiating a wormhole attack, an attacker overhears packets in one part of the network, tunnels them through the wormhole link (external to the network) to another part of the network. This  effectively generates a false scenario of the presence of the original sender in the neighborhood of the remote location. An illustration of a worm-hole is given in Figure \ref{fig:wormHoleDemonstration}.


Many routing algorithms depend on the nodes' ability to accurately discover their neighboring nodes. The nodes ordinarily perform a broadcasting beacons (including ID, and other information) to their neighbors. If the neighbor discovery beacons are tunneled through wormholes, the good nodes will get false information about their route. Although finding faulty routes is in itself a problem, worm holes can cause further critical security threats using these faulty routes. The resulting effect of wormholes on the routing is to include a worm hole link in most of the computed routes. This in turn, gives an attacker complete control of transmitting  great amounts of data, which may be selectively or completely dropped. Impacts of a wormhole on a route discovery procedure in a sensor network have been studied at length in \cite{Khabbazian,HuEvans}.

In the absence of known coordinates, as that provided for example by GPS, nodes in a sensor network depend on the positions of their neighbors to triangulate their own positions. A limited deployment of a hardware (GPS) over a few nodes would, on the other hand,  be sufficient for the entirety of the  nodes to compute their positions. Much work has been done on distributed localization in sensor networks \cite{shancangLi,JoseCosta}, and the predominant approach relies on strong correlation of geographic vicinity and communication capability of the nodes. Note that wormholes distort such correlation, and will hence adversely affect the localization algorithms. A study of impact of wormholes on localization procedures can be found in \cite{Khabbazian}.   In light of the serious impact worm holes may have on a sensor network, we propose to also naturally adapt the strategy we proposed for analyzing coverage problem to  not only detect but also localize these failures.

\subsection{\textbf{Paper Organization}}
\label{subsec:Organization}
The balance of the paper is organized as follows. In Section \ref{sec:Formalization}, we formalize both the coverage hole and the worm hole problem by a precise mathematical formulation. We subsequently provide the fundamental mathematical background necessary for topological analysis in Section \ref{sec:Background}. We provide a detailed discussion of our algorithm to localize the Coverage Hole in Section \ref{sec:CoverageProblem}, and describe its natural adaptation to the problem of the worm hole attack in Section \ref{sec:WormHoleProblem}. We conclude with some remarks in Section \ref{sec:Conclusion}.

\section{\large{\textbf{Formalization of topological network failures}}}
\label{sec:Formalization}
\subsection{\textbf{Coverage Problem}}
\label{subsec:FormalCoverageProblem}
We consider the scenario where $N$ sensor nodes are randomly deployed in a region of interest. We denote the collection of all the nodes as the set $V=\{v_i\}$. Each node $v_i$ can communicate with all the nodes within a circular neighborhood $R_c^i$ of radius $r_c^i$, and we denote these nodes as the set $\emph{N}(v_i)$, the neighbors of $v_i$. A communication graph $G=(V,E)$ is thus formed as the collection of the set $V$ together with the set of edges $E=\{(v_i,v_j)\}$ where $(v_i,v_j)\in E$, if and only if $v_i$, $v_j$ can communicate with each other. The coverage area of a sensor at each node is assumed to be a circular neighborhood $R_s^i$ or radius $r_s$ centered at the node $v_i$. Let $\Re$ denote the union of areas enclosed by the outermost boundaries of connected components of the network. The objective is to ensure that the following relation holds
\begin{equation}
\label{eqn:ProbStat}
\Re \subseteq \bigcup_i{R_c^i} = R_c,
\end{equation}
where $R_c$ is the total coverage space. This also highlights our interest in  $\Re$ being completely covered by the coverage areas of the sensors. The outermost boundary of a sensor network is to some extent in the control of the deployer, and there are algorithms which can detect this boundary \cite{BoundDetection}. As Equation (\ref{eqn:ProbStat}) suggests, we are therefore mainly interested in the coverage of the region "inside" the network. Furthermore, if the relation (\ref{eqn:ProbStat}) does not hold, our goal is to find the nodes which are closest to the boundary $\partial(\Re\setminus R_c)$ of the uncovered region. As an illustration of this problem, Figure \ref{fig:survivor1} shows a network with its communication graph and coverage area (the shaded region). The region of interest is the interior of the outermost boundary of the network. Since a part of this region is not covered by any sensor, we seek the smallest cycle in the network surrounding this coverage hole. We assume the following:

\begin{enumerate}
\item Let $Q$ be a clique in $G$, then
\begin{equation}
\label{eqn:ConvAssumption}
conv(Q) \subseteq \bigcup_{v_i \in Q} R_c^i,
\end{equation}
i.e., for any given clique $Q$ in the communication graph, the convex hull $\left(conv(Q)\right)$ of the nodes is completely covered. This assumption serves to  characterize the coverage area using the communication graph, as   further discussed  in Section \ref{sec:Background}. This can be ensured by requiring the relationship between the sensor coverage radius and the communication radius as  $r_s \geq \frac{r_c}{\sqrt{3}}$. Note that this assumption is not restrictive as the antenna power and hence the  communication radius, may be altered in order to extract the appropriate graph. This may be seen in Figure \ref{fig:survivor1}, where for any clique (for example, all the triangles), the interior is completely covered.
\item The nodes  have no  localization information.
\item There is no direction information, i.e., the nodes are unaware of  the relative orientation of their neighbors.
\item The nodes are not necessarily uniformly distributed in a given region of interest.
\end{enumerate}


\subsection{\textbf{Worm Hole Problem}}
\label{subsec:FormalWormHoleProblem}
A worm-hole attack is typically launched by two colluding nodes at positions $p_1$ and $p_2$ inside a network. Denote the neighborhood regions around these points by $N_1$ and $N_2$. The two attacking nodes may receive all the packets transmitted from within their respective neighborhoods, and relay them to the other. Denote by $V_1$ and $V_2$ the sets of vertices (sensor nodes) which lie in $N_1$ and $N_2$ respectively. The result of a worm-hole attack will be to produce a complete bi-partite graph with $V_1$ and $V_2$ as the two classes of vertices. The problem of localizing a worm hole attack, hence reduces to identifying the sets $V_1$ and $V_2$. In addition to all the above assumptions pertaining to the coverage problem, we will also assume the following:

\begin{enumerate}
\item The positions $p_1$ and $p_2$ are sufficiently far apart from each other inside the network. This assumption is based on expected topology change due to a wormhole and assumed detectability (this is possible only if $N_1\cap N_2 = \phi$).
\item The positions $p_1$ and $p_2$ are not very close to the outer boundary of the network.
\item The distribution of the nodes is sufficiently dense  such that a deletion of a node in the network will not cause a significant change in the path lengths.  This assumption is not necessary for detecting a worm hole, but is important for localizing the neighborhoods $N_1$ and $N_2$.
\end{enumerate}

Figure \ref{fig:wormHoleDemonstration} shows an example of a worm hole attack. In this case, $X$ and $Y$ are the positions $p_1$ and $p_2$ and the neighborhoods $A$ and $B$ are $N_1$ and $N_2$ according to our definition. Note that in the network shown, $N_1 \cap N_2 = \phi$, which will enable us to detect the attack. The assumptions 2 and 3 are however, not valid as $p_1$ and $p_2$ are close to the outermost boundary which violates assumption 2, and there is a bottleneck in the network which violates assumption 3. The algorithm presented in Section \ref{sec:WormHoleProblem} will cause some false alarms in this case. We further elaborate  on this in Section \ref{sec:WormHoleProblem}, and show some examples where we can accurately localize the attack.

\section{\large{\textbf{Framework}}}
\label{sec:Background}
The current state of research invokes areas from Mathematics such as Topology, Homological Algebra,  Engineering and Computer science (eg. gossip algorithms in Sensor networks and Graph theory). In this section, we construct a suitable framework for our algorithm, by introducing  the necessary mathematical and computational tools. While the available literature is extensive, we focus only on the important concepts which are central and  sufficient to elucidate the implications of our algorithm.

\subsection{\textbf{Topological Analysis}}
\label{subsec:TopologyBackground}
Topological analysis \cite{munkrees} can loosely be construed as the study of global organization of spaces without paying much heed to fine geometrical structure. For a space embedded in $\mathbb{R}^3$, this amounts to analyzing properties such as, ``is the space connected?'', ``does the space wrap upon itself?'', ``does the surface have any holes?'' or ``does the surface enclose a three dimensional void?'' and so on. As such, the developed tools provide the proper generalization to study organizational features of a  network, without expending resources on finer details. This generality is concisely captured in the notions of \emph{homotopic} mappings and \emph{homotopy equivalent} spaces, which are defined as follows:

\begin{definition}
Let $X$ and $Y$ be two spaces. Two maps $f_1,f_2:X\rightarrow Y$ are said to be \textbf{\emph{homotopic}} $(f_1\approx f_2)$ to each other if $\exists$ a continuous map $F:X\times I \rightarrow Y$ (where $I = [0,1]$) such that $F(s,0)=f_1(s)$ and $F(s,1) = f_2(s)$. Such a function $F$ is called a \textbf{\emph{Homotopy}} between $f_1$ and $f_2$.
\end{definition}

\begin{definition}
Two spaces $X$ and $Y$ are said to be \textbf{\emph{Homotopy Equivalent}} if $\exists$ continuous maps $f:X\rightarrow Y$ and $g:Y\rightarrow X$ such that $f\circ g \approx \textbf{id}$ and $g\circ f \approx \textbf{id}$. Such a map $f$ is called a \textbf{\emph{homotopy equivalence}}.
\end{definition}

The above definition means that if two spaces $X$ and $Y$ are homotopy equivalent, then one can be continuously deformed into the other, and they both have the same topological features. A remarkable result from algebraic topology is that the \emph{homology spaces} which we compute (described in Section \ref{subsec:homologicalAlgebraBackground}), are invariant to homotopic mappings. This is what enables us to treat a relatively large class of spaces in a unified framework without the costly and valuable resources required for considering their exact geometry. The computation of homology spaces does not depend on the localization information of nodes. Figure \ref{fig:homotopicSpaces} shows two homotopy equivalent spaces which may be viewed as coverage areas of two different sensor network. Although, their geometry (the distribution of distances between points) is quite different (this also reflects  the location of the nodes), they have the same topological features.
\begin{figure}[!h]
\centering
\includegraphics[width=0.3\textwidth]{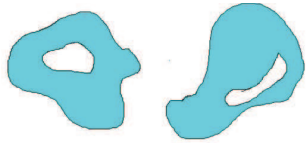}
\caption{Figure showing two homotopy equivalent spaces}
\label{fig:homotopicSpaces}
\end{figure}
We can further exploit this invariance to get homotopy equivalent representations of these spaces, using simple building blocks which, as  will be seen, are simple to manipulate. These representations are called \emph{simplicial complexes}. The simple building blocks are called simplices (simple pieces). The dimension of a simplex is represented by its order. Simplicial Complexes are representations of given topological spaces using simplices (simple pieces). A standard 0-simplex is just a point, a standard 1-simplex is a line segment, a 2-simplex a triangle and so on. A $k^{th}$ order simplex or $k$-simplex $\sigma^k$ is the set of all points given by the convex combination of $k+1$ linearly-independent points, $\sigma^k = (v_0,...v_{k})$.  Figure \ref{fig:simplices} shows simplices of order 0 through 3, and Figure \ref{fig:TopoSpaceSimplicialComplex} shows an example of representing a topological space using a simplicial complex. Note that the topology is preserved in the simplicial complex representation.
\begin{figure}[!h]
\centering
\subfigure[0-Simplex]{
\includegraphics[width=0.2\textwidth]{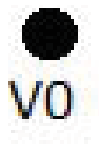}
\label{fig:simplex0}}
\subfigure[1-Simplex]{
\includegraphics[width=0.2\textwidth]{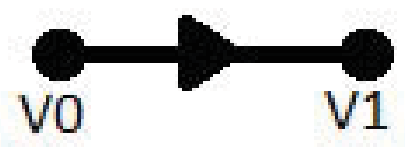}
\label{fig:simlex1}  }
\subfigure[2-Simplex]{
\includegraphics[width=0.2\textwidth]{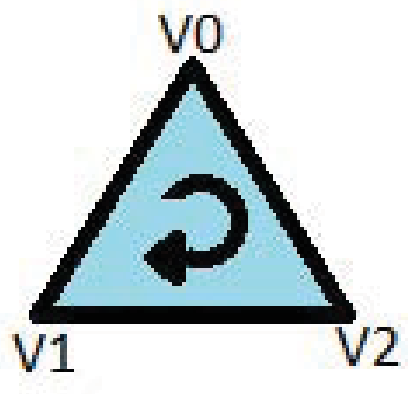}
\label{fig:simplex2}  }
\subfigure[3-Simplex]{
\includegraphics[width=0.3\textwidth]{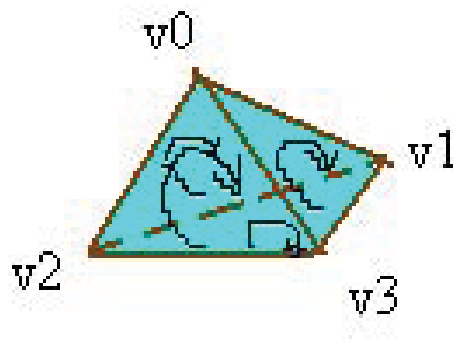}
\label{fig:simplex3} }
\caption{Simplices of order 1  to 4}
\label{fig:simplices}
\end{figure}
\begin{figure}[!h]
\centering
\subfigure[A topological space]{
\includegraphics[width=0.4\textwidth]{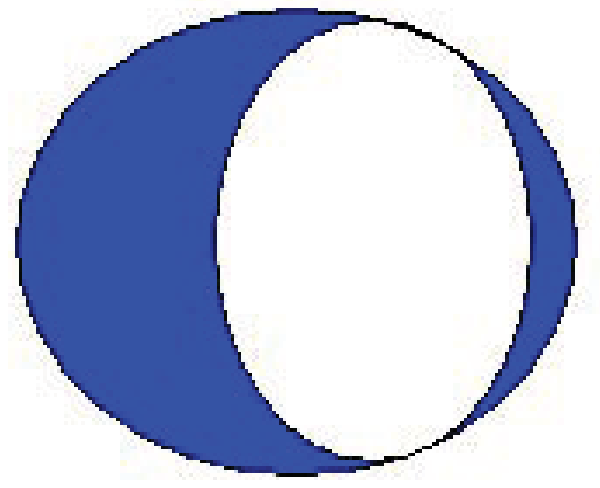}
\label{fig:ExampleTopoSpace} }
\hspace{5 cm}
\subfigure[Simplicial Complex representing the topological space]{
\includegraphics[width=0.5\textwidth]{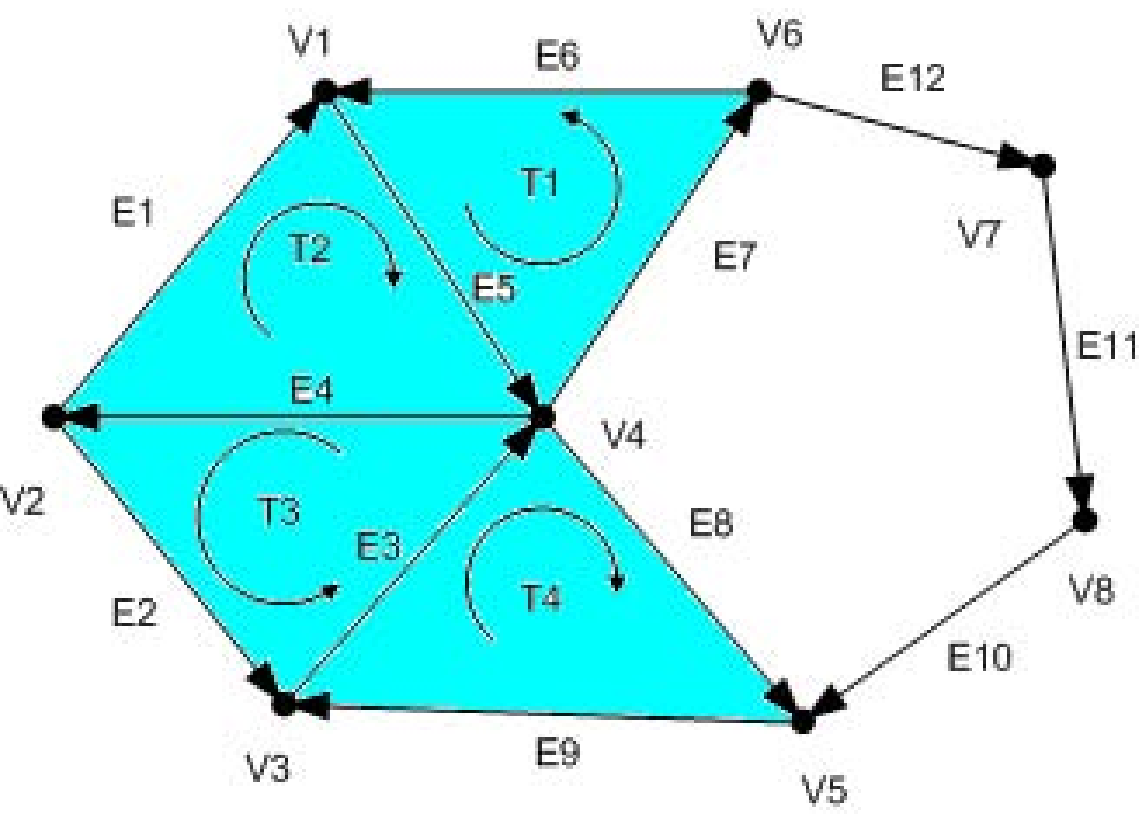}
\label{fig:homologyExample} }
\qquad \qquad
\caption{Representation of a Topological space using a simplicial complex}
\label{fig:TopoSpaceSimplicialComplex}
\end{figure}

\subsection{\textbf{Combinatrics}}
\label{subsec:combinatricsBackground}
Owing to its simple representation, a simplicial complex may be abstracted into a combinatorial object.  We can view this intermediate step as a transition from topological spaces into algebra for computing homology spaces. As each simplex is uniquely determined by specifying its vertices, given the total set of vertices $V=\{v_i\}$, a simplicial complex may then be abstractly specified as a collection of subsets $E_j \subset P(V), j = 1,2,\ldots$, where  $P(V)$ is the power set, and each $E_j$ is a collection of $j-$tuples from $V$ representing the $j-1$ simplices. Note that when we restrict $j$ to the set $\{1,2\}$, what we get is a graph. Therefore, a simplicial complex may be viewed as a generalization of a graph, and when it is homotopy equivalent to a space, it captures its topological properties. For a given $k$ simplex $\left(v_0,\ldots,v_k\right)$, we can also define an orientation by specifying the order of the vertices. We divide all possible permutations of these $k+1$ points into two classes, the elements of each class may be transformed from one to another by interchanging adjacent vertices an even number of times, giving a simplex two possible orientations. The simplices in Figures \ref{fig:simplices} and \ref{fig:TopoSpaceSimplicialComplex} show an orientation given to the simplices.\\
To get a representation of a coverage area, we can use a particular type of a simplicial complex called the \emph{nerve complex} \cite{ghristAbubakr}. Note that the coverage area  is a union of  convex sets.  Given a collection of sets $R_c = \bigcup_i{R_c^i}$ , the nerve complex (or the \^{c}ech complex) of $R_c$, $K_N\left(R_c\right)$, is the abstract simplicial complex whose $k$-simplices correspond to nonempty intersections of $k+1$ distinct elements of $R_c$. An edge in $K_N\left(R_c\right)$ exists between two vertices if and only if the corresponding elements of $R_c$ intersect. Higher dimensional
simplices are regulated by mutual intersections of collections of elements of $R_c$. Among the many uses of nerves in topology, the following classical result is perhaps of greater importance in applications:
\begin{theorem}
(The \^{C}ech Theorem): The nerve complex of a collection of convex sets has the homotopy type of the union of the sets.
\end{theorem}
The implication of this theorem is that $K_N\left(R_c\right)$ effectively captures the topology of $R_c$. The computation of the nerve complex unfortunately requires localization information, and is very difficult even when we have it. We therefore rely on an approximate representation called the Vietoris-Rips (or Rips in short) complex, denoted by $K_p\left(R_c\right)$ which can be obtained  from only the communication graph. For extracting the Rips complex, we simply say that each $k$ clique in the communication graph is a $k-1$ simplex in the Rips Complex. Under assumption (1) in Section \ref{subsec:FormalCoverageProblem}, the Rips complex is a reasonable approximation of the underlying topological space in the sense that the number of false alarms and false negatives indicating coverage holes are very low in number. The reader is refered to \cite{RghristHomo} for examples where the Rips complex does not accurately represent the coverage area. We however maintain that this is not a limitation and as shown in \cite{ghristAbubakr}, we can always represent $R_c$ accurately using two Rips complexes using appropriate radius of communications. In particular, for a coverage radius $r_c$, the authors show that, if we have communication radii (strong and weak) $r_1$ and $r_2$ such that $2r_c = r_1 \geq \sqrt{2}r_2$,  then the rips complexes $K_p^{r_1}\left(R_c\right)$ and $K_p^{r_2}\left(R_c\right)$ satisfy the following relation:
\begin{equation}
K_p^{r_2}\left(R_c\right) \subset K_N\left(R_c\right) \subset K_p^{r_1}\left(R_c\right)
\end{equation}
The above relation implies that the topology of $R_c$ is completely captured by the two rips complexes. This is tantamount to using our algorithm twice. For the purpose of this paper, we will assume that the Rips complex obtained with condition (1) in Section \ref{subsec:FormalCoverageProblem} accurately represents the Coverage area. We will describe in the following section, an approach to infer the topological properties of $R_c$ using Algebra on $K_p\left(R_c\right)$.

\subsection{\textbf{Homological Algebra}}
\label{subsec:homologicalAlgebraBackground}
In this section, we discuss some fundamental notions of homology spaces. We subsequently relate the algebraic structure to the combinatorial structure of the Rips Complex, and demonstrate the usefulness of these spaces in inferring the existence and the cardinality of coverage holes.
\begin{definition}
A \textbf{\emph{Chain Complex}} $\{C_k,\partial_k\}$ is a sequence of vector spaces $\{C_k\}$ together with linear operators $\{\partial_k:C_k\rightarrow C_{k-1}\}$ called the boundary operators,
$$ \rightarrow C_n \overset{\partial_n}{\rightarrow}C_{n-1}\overset{\partial_{n-1}}{\rightarrow}\cdots \overset{\partial_{k+1}}{\rightarrow}C_k\overset{\partial_k}{\rightarrow}C_{k-1}\cdots \overset{\partial_1}{\rightarrow} C_0\rightarrow0 $$
with the boundary operators satisfying
\begin{equation}
\label{equ:boundaryOperatorCondition}
\partial_{k-1} \circ \partial_k = 0 \mbox{  or  } \partial^2 = 0
\end{equation}
The groups $\{C_k\}$ are called \textbf{\emph{chain spaces}} and their elements are called \textbf{\emph{chains}}.
\end{definition}
The chain complex is fundamental to homological algebra, as it provides the structure where homology spaces may be defined.
Note that since $\partial_{k-1} \circ \partial_k = 0$,  it follows that the image of one boundary operator is a subset of the kernel (or null space) of the next boundary operator, i.e.,
$$ Img(\partial_k) \subset Ker(\partial_{k-1}) $$
This observation enables us to define a homology space as follows,
\begin{definition}
\label{def:Homology}
Given a chain complex $C = \{C_k, \partial_k\}$, the $k^{th}$ \textbf{\emph{homology group}} $H_k(C)$ of the chain complex is given as
\begin{equation}
\label{equ:Homology}
 H_k(C) = ker(\partial_k)/Img(\partial_{k+1})
\end{equation}
i.e., the $k^{th}$ homology group is the quotient group formed by equivalent classes of elements in $ker\left(\partial_k\right)$, where the elements are considered equivalent if their difference lies in the subspace $Img\left(\partial_{k+1}\right)$.
\end{definition}
Of particular interest to us, is the first homology space $H_1(C)$. We will hence work only with $C_2$, $C_1$  and $C_0$. We form the chain spaces $C_2,C_1$ and $C_0$ by taking all the 2-simplices (triangles), 1-simplices (edges) and vertices respectively of the Rips Complex $K_p\left(R_c\right)$ as the basis vectors. The additive inverses in the chain spaces are given in terms of the orientation as:
\begin{equation}
    \label{equ:SimplexSign}
    \mbox{if } \sigma^k = (v_0,\ldots,v_i,v_{i+1},\ldots,v_{k}) \mbox{  then  } -\sigma^k = (v_0,\ldots,v_{i+1},v_i,\ldots,v_{k})
\end{equation}
and the boundary operator is defined in terms of $k$-dimensional simplices as:
\begin{equation}
\label{equ:BoundaryOperator}
 \partial_k(v_0,\ldots,v_k) = \sum_i{-1^i(v_0,\ldots,v_{i-1},v_{i+1},\ldots,v_k)}
\end{equation}
It is simple to check that the boundary operator so defined satisfies Equation (\ref{equ:boundaryOperatorCondition}), and we show this fact here by considering, for example, the action of $\partial_1 \circ \partial_2$ on a two simplex $(v_0,v_1,v_2)$.
\begin{equation}
\label{equ:del2iszero}
\partial_1 \circ \partial_2(v_0,v_1,v_2) = \partial_1((v_1,v_2) - (v_0,v_2) + (v_0,v_1)).
 = v_2 - v_1 -v_2 + v_0 + v_1 - v_0 = 0
\end{equation}
Using the above definitions of boundary operators and chain spaces, we can form a chain complex $C\left(R_c\right)$ using the combinatorial structure in the Rips Complex.\\
In order to understand what homology groups tell us about the topological space, we should carefully look at the action of the boundary operators.  Let us look at the null space (kernel) of $\partial_1$. Consider a cycle $c=e_1+e_2+e_3+e_3$ as shown in Figure (\ref{fig:boundaryKernelExample}) which is homotopic to a loop. The action of $\partial_1$ is given as:
$$ \partial_1(c) = v_2-v_1+v_3-v_2+v_4-v_3+v_1-v_4 = 0$$
This implies that the null space of $\partial_1$ consists of all closed cycles (chains without boundaries).  And as we saw in Equation \ref{equ:del2iszero}, the boundaries of $k+1$-simplices are closed cycles in $C_k$, and they belong to $ker(\partial_k)$.  This means that $ker(\partial_1)$ also consists of closed cycles which are boundaries of 2-simplices. But we know that 2-simplices are homeomorphic to disks or any space without any holes in them. Therefore, if we remove all the cycles which are boundaries of 2-simplices, the cycles that remain are those circling a hole. From the definition of the homology group $H_1\left(C\left(R_c\right)\right) = ker(\partial_1)/Img(\partial_2)$, it is clear that $H_1$ counts the number of holes in our topological space.We now present an example to illustrate the basic mechanism of this procedure

\begin{figure}[!h]
\centering
\includegraphics[width=0.4\textwidth]{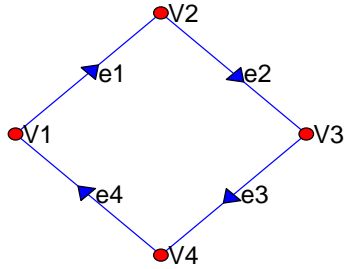}
\caption{a chain in $ker(\partial_1)$}
\label{fig:boundaryKernelExample}
\end{figure}

\subsubsection{\textbf{Example}}
\label{subsubsec:Example}
Consider the simplicial complex $X$ shown in Figure \ref{fig:homologyExample}. The orientation of the simplices (1 and 2-dimensional) are arbitrarily chosen; the homology assigned spaces  will be independent of this choice. Consider the 1-chains (paths): $c_1$, the outermost boundary, $c_2$, the closed path enclosing the triangles and $c_3$, the closed path enclosing the hole, all in clockwise orientations. These chains in terms of the basis vectors (the simplices) are expressed as
$$c_1 = E1 - E6 + E12 + E11 + E10 + E9 + E2,$$
$$c_2 = E1 - E6 - E7 + E8 + E9 - E2, \hspace{1 cm}c_3 = E12 + E11 + E10 - E8 + E7.$$
Note that Equation (\ref{equ:SimplexSign}) states that a change in the sign of a simplex changes its orientation.  Using Equation (\ref{equ:BoundaryOperator}) for the boundary operator, we can see that
\\$\partial_1(c_3) = \partial_1(E12) + \partial_1(E11) + \partial_1(E10) - \partial_1(E8) + \partial_1(E7) = V7 - V6 + V8 - V7 + V5 - V8 - (V5 - V4) + V6 - V4 = 0$
\\Similarly, any closed path (including $c_1$ and $c_2$) can be shown to belong to $ker(\partial_1)$. Again using Equation (\ref{equ:BoundaryOperator}), the action of the $\partial_2$ operator, on $T1$ for example, is given as
\\$\partial_2(T1) = E6 + E5 + E7$.
\\It should be easy to verify that $c_2$ can be expressed as $c_2 =\partial_2( T4 - T3 + T2 - T1)$ and therefore, $c_2 \in img(\partial_2)$ Note also that $c_1 - c_3 = c_2$, i.e., $c_1$ and $c_3$ differ by a chain in $img(\partial_2)$ and are therefore, homologous. In other words, they encircle the same hole. To compute $H_1(X)$, first observe that any closed path on $X$ may be expressed as a sum of the closed paths surrounding the four triangles and, that surrounding the hole. Therefore, $ker(\partial_1)$ is a vector space with 5 basis vectors, i.e., $ker(\partial_1) \cong \mathbb{R}^5$.  Also, the four closed paths generated by the action of $\partial_2$ on the four triangles (the basis vectors for $C_2$) are linearly independent and therefore, $img(\partial_2) \cong \mathbb{R}^4$. From definition (\ref{equ:Homology}) of a homology space as a quotient space, we can see that $H_1(X) \cong \mathbb{R}$. The first homology group has one generator, corresponding to one hole in the complex.

\subsubsection{\textbf{Laplacians}}
\label{subsubsec:LaplaciansBackground}
As we saw in the above example, the computation of the dimension of the $H_1$ (the first \emph{betti} number) involves computing the ranks of the operators $\partial_1$ and $\partial_2$. Such a task is computationally very expensive, and as we will show in Section \ref{subsec:coverageHoleDetection}, the precise rank of these operators is not necessary for detecting the existence of a hole. Laplacian operators provide an easier way to detect the triviality of the homology spaces. The graph laplacian  from graph theory may be generalized to the case of simplicial complexes \cite{controlLap} as
\begin{definition}
Given a chain complex $C$, the $k^{th}$ Laplacian operator $L_k:C_k\rightarrow C_k$  is defined as
$$ L_k = \partial_{k+1}\circ\partial_{k+1}^{\ast} + \partial_{k}^{\ast}\circ\partial_k $$
where $\partial_k^{\ast}$ is the adjoint of $\partial_k$
\end{definition}
It may be shown \cite{controlLap} that the kernel of the laplacian operator $L_k$ is isomorphic to the $k^{th}$ homology group, i.e., $ker(L_k) \cong H_k(C)$, and we can use the laplacian operators to infer the topology of $K_p\left(R_c\right)$. An important property of the Laplacian operators is that they are symmetric and non-negative definite.

\subsection{\textbf{Distributed Computation}}
\label{subsec:DistributedComputationBackground}
In this Section, we address issues central to sensor networks, chief among them the scaling of computation with the network size, and the implementation of related mathematical tools.\\
Owing to excessive cost of communication between nodes , gathering all the raw data at the nodes to a sink node is prohibitive. Whenever possible, distributive algorithms should be designed to reduce the demand for data collection. The power consumption during communication is in addition higher relative to that required for computations within the nodes, thus highlighting the importance for algorithms to reduce communication by with-in node computations. The use of positioning systems such as GPS or other localization algorithms, is also very expensive, emphasizing the use of localization information be avoided if at all possible. The algorithm we propose here satisfies all these basic requirements. \\
An interesting class of distributed algorithms is that of \emph{gossip algorithms}, where nodes process the data by passing messages amongst their neighbors. One particular gossip algorithm we exploit extensively, is the distributed computation of eigenvalues of sparse matrix in a network \cite{Kempe} using the orthogonal (or power) iteration method. In particular, as described in Section \ref{subsec:coverageHoleDetection}, we wish to compute the spectral radius of the first order laplacian $L_1$ of the Rips complex $K_p\left(R_c\right)$. We can extract the Rips Complex from the communication graph, distributively compute $L_1$ and its spectral radius as described in \cite{decentralizedComputationHomology}. We also develop some simple gossip algorithms and prove their efficacy in the following sections as and when required.

\section{\large{\textbf{Coverage Hole Localization}}}
\label{sec:CoverageProblem}

\subsection{\textbf{Algorithm Overview}}
\label{subsec:CovAlgoOverview}
In this section, we present a novel method to reduce the problem of \emph{locating} a coverage hole in a network into one of \emph{detecting} a hole using a "divide and conquer" mehtod. As seen in the previous section, the problem of detecting a hole in $R_c$ reduces to checking whether the first homology space of the chain complex formed from $K_p\left(R_c\right)$ is trivial, i.e., if $H_1\left(K_p\right) \cong 0$, there are no holes in the space. We subsequently and strategically divide the network into smaller partitions, and check for the presence of  holes in each of these partitions, and "drop" the partitions where there are none. As we continue this process, the partitions which survive, are the boundaries of the holes. The crux of this algorithm lies in the process of dividing the network in such a way so as to preserve the topology, i.e., neither create nor destroy holes.  Table \ref{Tab:AlgoOverview} presents an overview of this algorithm.

\begin{table}[!ht]
\label{Tab:AlgoOverview}
\centering
\begin{tabular}{l}
\hline
for all partitions with non-trivial homology\\
\hspace{0.35in}$\setminus\setminus$ Dissection\\
\hspace{0.35in}\emph{step 1:} Find the diameter nodes.\\
\hspace{0.35in}\emph{step 2:} Find the boundary nodes and construct the partitions.\\
\hspace{0.35in}$\setminus\setminus$ Detect holes in each of the above two partitions\\
\hspace{0.35in}for the two partitions constructed\\
\hspace{0.5in}\emph{step 3:} Compute the Laplacian Matrix $L$\\
\hspace{0.5in}\emph{step 4:} Check the rank deficiency of $L$\\
repeat.\\
\hline
\end{tabular}
\caption{Hole Localization Overview}
\end{table}

\subsection{\textbf{Hole Detection}}
\label{subsec:coverageHoleDetection}
In order to determine the number of holes in $R_c$, we have to compute the dimension of $ker(L_1)$. If on the other hand, the detection of a hole is only of interest, we may check whether $L_1$ is rank deficient, or in other words,  to check whether $L_1$ has a zero eigenvalue. To that end, we use the following theorem:

\begin{theorem}
\label{theo:RankDeficiency}
Let $L_1$ be a symmetric non-negative definite matrix with spectrum $\sigma(L_1)$, and spectral radius $\rho(L_1)$. Then, $L_1$ is rank deficient if  $\rho(\rho(L_1)I - L_1) = \rho(L_1)$
\end{theorem}

\begin{proof}
Let $\mathbf{x}$ be an eigenvector corresponding to the eigenvalue $\lambda \in \sigma(L_1)$. Then $(\rho(L_1)I - L_1)\mathbf{x} = (\rho(L_1)-\lambda)\mathbf{x}$. $\Rightarrow \, x$ is also an eigenvector of $(\rho(L_1)I - L_1)$ and its eigenvalue  is $\rho(L_1)-\lambda$. Furthermore, $L_1$ is non-negative definite $\Rightarrow \, \lambda \geq 0 \, \Rightarrow \rho(\rho(L_1)I - L_1) \leq \rho(L_1)$
if $L_1$ is of full rank, then $\lambda > 0$
 $\Rightarrow \rho(\rho(L_1)I - L_1) < \rho(L_1)$.
 \qed
\end{proof}

The spectral radius of $L_1$ can be computed using the power iteration method by searching for the largest eigenvalue which can be distributively carried out over the network \cite{Kempe}. The convergence of the power iteration method (for eigenvector) is slow when the difference between the largest and second largest eigenvalue is small, while the eigenvalue itself quickly converges to the true value. A false detection of a hole is possible  when the smallest eigenvalue is very close to zero, but this problem is unlikely to happen in successive partitions (partitions are explained in the next section). Each iteration in the power iteration method includes multiplying $L_1$ by a vector from the previous iteration, and normalizing the resulting vector. The sum of the squared elements of the vector (for normalization) can also be distributively computed  in the network by a gossip algorithm whose convergence time is of the order $\Theta\left(n\log(n)\right)$ \cite{RandomWalks}\cite{RGA}.

\subsection{\textbf{Hole Localization}}
\label{subsec:HoleLocalization}
Each element in the first homology space $H_1$ represents an equivalence class of homologous closed paths encircling a hole in the coverage space. As such, \emph{Localizing} the exact boundary of this hole is in essence a problem of finding the smallest closed path in an equivalence class. A very direct approach was proposed in \cite{DistLocHoles}, where the authors formulate the localization as an optimization problem to seek the sparsest chain in the $H_1$ space. Such an approach is effective at the cost of a very slow convergence, and involves all the nodes in the network to participate in the optimization. While the presence of holes in a coverage space is a global property, the boundary of a hole is constrained to a  relatively small part of the network. Any ability to \emph{detect} a hole in this region in noway depends on the configuration of nodes in other parts of the network. We exploit this idea to reformulate the problem of \emph{identifying} a boundary of a hole to a much simpler problem of \emph{detecting} holes.
\\We accomplish this reformulation by iteratively dissecting the network into two smaller partitions, and by detecting the presence of holes in these smaller partitions. All nodes in the partition where no hole is detected, go into  a ``\emph{sleep}'' mode and are taken out of the analysis, yielding a  a  valuable power saving. The remaining active nodes will form partitions with non-trivial homology (with holes in coverage), to get further dissected in pursuit of hole localization. We will thus be rapidly converging onto the exact boundary of the holes  with each iteration. In the first iteration, each connected component of the network graph $G$ is treated as a partition. The partitioning strategy is to minimize the "size" of the resulting partitions, while simultaneously preserving the overall topology.

\subsubsection{\textbf{Finding Diameter Nodes}}
\label{sec:FindingDiameterNodes}
Firstly, we elaborate on what we mean by ``size'' in the above description. The time required to complete steps 1 and 2 above, directly depends on the \emph{diameter} of the network partition. Step 4 utilizes a gossip algorithm whose convergence does depend on the diameter, with other factors possibly  coming into play \cite{GossipAlgo}. A network segmentation obtained by minimizing the diameter of the smaller partitions is therefore optimal for minimizing the overall run time. This is facilitated by identifying a pair of nodes called the \emph{diameter nodes} defined by
\begin{equation}
\label{eqn:diameterNodes}
\left(\bar{u},\bar{v}\right) = \arg\max_{(v_i,v_j)}\, d\left(v_i,v_j\right),
\end{equation}
where $d\left(v_i,v_j\right)$ is the shortest path between nodes $v_i$ and $v_j$ (in terms of hop count) in the partition of interest. Such a pair will generally not be unique, and ``ties'' between nodes are broken by a simple protocol which chooses the pair that has the node with the smallest ID.
\\We determine the boundary nodes in two stages; We first find  the candidate nodes $C_{dia}$ by assigning a scalar field $f(v_i)$ equal to the farthest distance for each node $x$ in the current partition, and to ultimately select the nodes with the maximum $f$;  we subsequently proceed to break the ``ties'' using the afore mentioned criterion,
\begin{eqnarray}
\label{eqn:CandidateDiameter}
f(v_i) = \max_{v_j} \, d\left(v_i,v_j\right) \nonumber\\
C_{dia} = \{v| f(v) = \max_{v_j} \, f(v_j)\}.
\end{eqnarray}
To compute $f$ on $G$, we use a simplified version of the \emph{Dijkstra's} algorithm. The simplification is a result of the following differences with  Dijkstra's, a)we do not need the shortest paths but rather just the distances and b)Instead of shortest distance from a node $v_i$ to all other nodes, we require \emph{max} of distances.

\subsubsection{\textbf{Computing the scalar field $f$}}
\label{sec:Computingf}
A summary of the algorithm for computing $f$ is given in Table \ref{Tab:DiameterNodesAlgo}. In what follows, we  provide an intuition into the mechanics of the algorithm followed by a mathematical justification.
\\It immediately follows from Equation (\ref{eqn:CandidateDiameter}) that, in order to compute $f(x)$, it is sufficient for each node $v_i$ to have the knowledge of $d\left(v_i,v_j\right)$ for all $v_j$ in $G$. Since $f(v_i)$ has to be computed for all $v_i$, the preceding statement may  equivalently be stated as; for each $v_i$, it is sufficient for all $v_j$ (all other nodes) to know $d\left(v_i,v_j\right)$. We accomplish this by broadcasting node $v_i$'s id in the network, and for each node $v_j$, $d\left(v_i,v_j\right)$ is equal to the number of hops taken by the first message arriving at $v_j$. Note that, in order for $v_i$'s id to reach all other nodes (assuming a connected Graph), it is sufficient that any other node broadcasts this information to its neighbors only once, since re-broadcasting will provide no new information. This will result in reducing the number of required broadcasts. In order to ensure that no message (id) is re-broadcast, it is sufficient for each node to remember all the messages it previously transmitted (for example, by maintaining a table). The next theorem assuages this requirement by showing that it is sufficient for a node to remember all the messages only for a limited time. This reduces the memory requirement on the nodes. We now provide a mathematical justification for the above intuitive arguments.
\\Denote as $A = \{a_{ij}\}$, the adjacency matrix for $G$, then $A^n = {a_{ij}^n}, n>0$ where $a_{ij}^n$ is the number of paths of length $n$ from $i$ to $j$ \cite{GodsilRoyle}. For simplicity, we assume that $i$ is the ID given to node $v_i$.  Now, the shortest distance from $v_i$ to $v_j$, $i\neq j$ is given by
\begin{equation}
\label{eqn:shortestDistance}
d(v_i,v_j) = \arg\min_{n>0} \, a_{ij}^n > 0, i\neq j
\end{equation}
The matrix $A^n$ can be distributively represented  in the network where node $v_i$ computes and stores the $i^{th}$ row. This can be iteratively computed  as $A^{n+1} = A\cdot A^n$, and $a_{ij}^{n+1}$ at node $v_i$ is obtained as $a_{ij}^{n+1} = \sum_{v_k\in N(v_i)}{a_{kj}^n}$. This computation is enabled by all the nodes broadcasting their row to their neighbors. If $m$ is the smallest integer such that $a_{ij}^m>0$, this implies there is no path from $v_i$ to $v_j$ of length smaller than $m$. Therefore, the node $i$ ``discovers'' node $j$ at iteration $m$ and at this instant, is a "new" node. Further, if $k\in N(i)$, this also implies $a_{kj}^{m+1}>0$ and the values of $a_{kj}^n$ for $n>m+1$ are irrelevant from the perspective of computing $d\left(v_k,v_j\right)$. We therefore refrain from broadcasting $a_{ij}^m$ for $n>m$ time intervals. In other words, each node broadcasts the information about a new node it discovers only once. Note that in so doing, we do not actually compute $A^n$ at the $n^{th}$ iteration, but an estimate $\hat{A}^n$ with the property that the smallest integer $m$ for which $\hat{a}_{ij}^m > 0$ is that for $a_{ij}^m>0$. The table used, acts as  a reference to avoid transmitting duplicate information to its neighbors. Here, it appears that the memory required at each node will be equal to the number of nodes in the partition, as all the nodes will eventually be  discovered. We maintain that it suffices to store a node in the table for only two iterations.\\
\begin{theorem}
\label{theo:StorageLength}
A node $v_i$ storing the information about the node $v_j$ for two iterations, guarantees no duplicate information is broadcasted.
\end{theorem}
\begin{proof}
By contradiction.
\\Duplicate information will be broadcast if a node $v_i$ discovers a node $v_j$ at iterations $m$ and $m+t, t>2$. This means that there are two paths $P = (j,p_1,\ldots,p_{m-1},i)$ and $Q = (j,q_1,\ldots,q_{m+t-1},i)$. At the $m^{th}$ iteration, node $v_i$ will start a broadcast which propagates along $Q$ in the reverse direction and meets the message coming along $Q$ at node $q_{m+t_1}=q_{m+t-t_1}$.   $\exists t_1 = (t-1)/2 \mbox{ or } (t-2)/2$, whichever is an integer such that $(m+t-t_1)-(m+t_1)=1\mbox{ or  } 2$. The message from $v_j$ would take the path $Q$ only if node $q_{m+t_1}$ broadcast it, thus violating the rule because $v_j$ is already in the table at that instant.
\qed
\end{proof}
If $n_{max}$ is the largest distance for which a node $v_i$ discovers a new node, then we set $f(v_i)=n_{max}$.

\subsubsection{\textbf{Diameter nodes}}
\label{sec:DiamterNodes}
Once $f$ is computed, candidate diameter nodes are found by consensus for maximizing $f$ on the network by a simple gossip algorithm. There are many algorithms in the literature for computing such aggregates on the network, for example \cite{Aggregates}. The essence of such algorithms is that at each iteration, if a node "discovers" a new max value, it broadcasts this discovered value to all its neighbors.  Similarly, the diameter nodes are obtained from the candidate nodes by consensus for a minimum of node IDs.

\begin{table}
\label{Tab:DiameterNodesAlgo}
\centering
\makebox[0.4\textwidth][l]{
\begin{tabular}{l}
\hline
\bfseries At each Node $i$ in the segment\\
\hline
$\setminus\setminus$ \textbf{Computing $f$}\\
$\setminus\setminus$ Initialization: Discover itself\\
add $v_i$ to table and broadcast to $N(i)$\\
$\setminus\setminus$ run time\\
at iteration n:\\
\hspace{0.25in}$\setminus\setminus$check for new nodes discovered\\
\hspace{0.25in}\emph{if} found new nodes\\
\hspace{0.5in}broadcast new nodes to $N(i)$\\
\hspace{0.5in}add new nodes to table\\
\hspace{0.5in}clear values of n-2 iteration\\
\hspace{0.25in}\emph{else}\\
\hspace{0.5in} $f(v_i)=n$\\
\hspace{0.5in} stop.\\
\hline
\end{tabular}
}
\\
\makebox[0.4\textwidth][r]{
\begin{tabular}{l}
\hline
\bfseries At each Node $i$ in $V_X$\\
\hline
$\setminus\setminus$Initialization\\
\emph{if} $v_i$ is a diameter node\\
\hspace{0.25in} broadcast $i$ to $N(i)$. stop.\\
\hspace{0.25in} ($i$ will serve as the segment ID)\\
\emph{else}\\
\hspace{0.25in} wait until reception\\
\hspace{0.25in} \emph{if} received two distinct IDs\\
\hspace{0.5in} broadcast the lowest received ID to $N(i)$.\\
\hspace{0.5in} $v_i$ = boundary node.\\
\hspace{0.25in} wait one time interval\\
\hspace{0.25in} \emph{if} received two distinct IDs overall\\
\hspace{0.25in} $v_i$ = boundary node.\\
stop.\\
\hline
\end{tabular}
}
\caption{\textbf{Finding Diameter Nodes}  and \textbf{Boundary Nodes} }
\end{table}

\subsubsection{\textbf{Finding Boundary Nodes}}
\label{sec:FindingBoundaryNodes}
As the physical positions of the nodes do not change, we form a virtual segmentation by finding boundary nodes $B = \{b_i\}$ within a partition which stop messages from passing through. This effectively separates a given partition into two parts with non-intercommunicating nodes. For a set $B$ to behave like a boundary \footnote{This definition of a boundary should not be confused with the conventional notion, confounded with the closure of a graph or a region. The particular definition we are using will be clear from the context.}, it has to satisfy certain properties:
\begin{definition}
Let $X=(V_X,E_X)\subseteq G$ be a connected sub-graph. The set of nodes $B$ is said to be a boundary in $X$, if and only if $\exists$ two disjoint sets $V_{X1},V_{X2}\subset V_X$ such that there is a node $b_i\in B$ in any path $(v_i,\ldots,v_j)$, where $v_i\in X_1$ and $v_j \in X_2$. Furthermore, $V_{X1}\cup V_{X2} \cup B = V_X$.
\end{definition}
If every path from $V_{X1}$ to $V_{X2}$ contains a boundary node, this means there is no path along which a message from $V_{X1}$ can reach $V_{X2}$, thus virtually separating both. This justifies the above definition for the boundary. The boundary nodes identify their neighbors as belonging to $V_{X1}$ or $V_{X2}$, and do not transmit messages from one to the other. \\
To minimize the diameter of the resulting partitions ($S_1=V_{X1}\cup B$ and $S_2=V_{X2}\cup B$), we choose the boundary nodes to be equidistant from the determined diameter nodes. This will cause  the boundary nodes to bisect the diameter of $X$. These equidistant nodes are obtained using a simple flooding algorithm which is presented in Table \ref{Tab:DiameterNodesAlgo}. The basic idea is to start a flood from both  diameter nodes, and determine the boundary nodes where these floods meet. Every node will either belong to $S_1$ or  to $S_2$ since $X$ is connected and therefore, $S_1 \cup S_2 = X_1 \cup X_2 \cup B = X$. Let the diameter points be $x_1$ and $x_2$, and let $v_1 \in X_1$ and $v_2 \in X_2$. This implies $v_1$ and $v_2$ received a single ID, ID$(x_1)$ and ID$(x_2)$ respectively. It follows that for any path $p = (v_1,\ldots,v_2), \exists v_i \in p$ such that $v_i$ received both  IDs and hence belong to $B$. This shows that the nodes obtained as in Table \ref{Tab:DiameterNodesAlgo}, indeed satisfy the definition of the boundary.\\

\begin{figure}[!h]
\centering
\includegraphics[width=.6\textwidth]{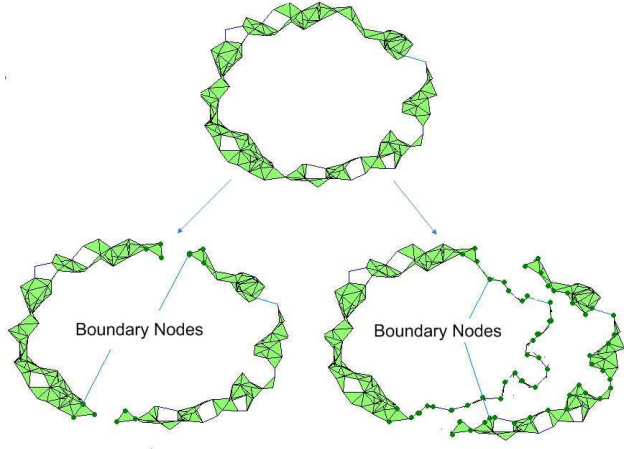}
\caption{The partitioning on the left does not preserve topology whereas that on the right does.}
\label{fig:topologyPreserve}
\end{figure}
An additional and very important property that a boundary partitioning should satisfy, is that  it should preserve the topology of the original entity (see Figure \ref{fig:topologyPreserve}). Specifically, if $X$ has no holes in its coverage, then neither $S_1$ nor $S_2$ should; and if $X$ has a hole, then it should be preserved in either one of the partitions. Theorem \ref{theo:TopologyPreservation} shows that a sufficient condition for preserving the  topology is the contractibility of the Rips complex obtained from an induced subgraph on the boundary nodes $B$. The Rips complex for $X$ is obtained by taking all the cliques as simplices, and similarly for the subgraphs induced on $S_1$, $S_2$ and $B$.\\

\begin{theorem}
\label{theo:TopologyPreservation}
Let $X$ be a Simplicial complex and $A,B\subset X$ be sub complexes such that $A\cap B$ forms a boundary on the underlying graph. Then $H_1(X) = H_1(A)\oplus H_1(B)$ if $A\cap B$ is contractible, i.e., $H_0(A\cap B) = \mathbb{R}$ and $H_1(A\cap B) = 0$.
\end{theorem}
\begin{proof}
For any simplicial complex $X$, and $A,B \subset X$, $\exists$ the following exact sequence called the Meyer-Vietoris Exact sequence.
$$H_1(A\cap B) \overset{\phi}{\rightarrow} H_1(A) \oplus H_1(B) \overset{\psi}{\rightarrow} H_1(X)$$
\\where $\phi$ and $\psi$ are linear operators. Now,
\\$H_1(A\cap B) = 0 \Rightarrow img(\phi) = 0 \Rightarrow ker(\psi) = 0 \Rightarrow$ $\psi$ is injective since it is linear. Therefore, $H_1(A) \oplus H_1(B) \subseteq H_1(X)$
\\Let $c\in ker(\partial_1(C_1^X))$ be a chain in the null space of the first boundary operator acting on the first chain space of $X$ ($c$ is a closed path), such that it contains $v_1\in A$ and $v_2 \in B$. $\exists b_1, b_2 \in c, b_1 \neq b_2$ such that $b_1,b_2$ also $\in A\cap B$ since $A\cap B$ is a boundary. (See Figure \ref{fig:digramForTheorem}). Now, since $A\cap B$ is connected, $\exists$ a chain corresponding to the path $b_1\rightarrow b_2$. Consider the two chains corresponding to closed paths $c1 := (v_1\rightarrow b_1 \rightarrow b_2 \rightarrow v_1)$ and $c2 := (v_2\rightarrow b_2 \rightarrow b_1 \rightarrow v_2)$. It immediately follows that $c1 + c2 = c$. Therefore, any chain in $ker(\partial_1(C_1^X))$ can be expressed as a sum of chains in $ker(\partial_1(C_1^A))$ and $ker(\partial_1(C_1^B))$, $\Rightarrow H_1(X) \subseteq  H_1(A) \oplus H_1(B)$.
\\ $\Rightarrow H_1(X) = H_1(A) \oplus H_1(B)$.
\qed
\end{proof}

\begin{figure}[!h]
\centering
\includegraphics[width=0.5\textwidth]{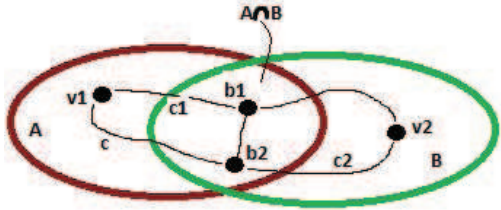}
\caption{Figure illustrating a chain in $X$ can be represented as a sum of chains in $A$ and $B$ where $A,B \subseteq X$ and $A\cap B$ is connected}
\label{fig:digramForTheorem}
\end{figure}

The first part of the theorem states that no new holes are created by the partitioning, and the second states that all the holes are preserved. If the boundary nodes obtained by the algorithm given in Table \ref{Tab:DiameterNodesAlgo} are not connected, we can form a tree by joining different connected components by a shortest path between them. This shortest path can be discovered by a simple flooding in the network originating at the connected components. The boundary obtained is also usually contractible, aside from one exception. As shown in Figure \ref{fig:exceptionFigure}, this happens exactly when $d(x_1,v_1) = d(x_1,v_2) = d(x_2,v_3) = d(x_2, v_4)$ in the given configuration, where $x_1$ and $x_2$ are the diameter nodes. In this case, all the nodes $v_1,v_2,v_3,v_4$ will be made boundary nodes. Note that we use the contractibility condition only to prove that no new holes are created, which is clearly also valid in this case.

\begin{figure}[!h]
\centering
\includegraphics[width=.4\textwidth]{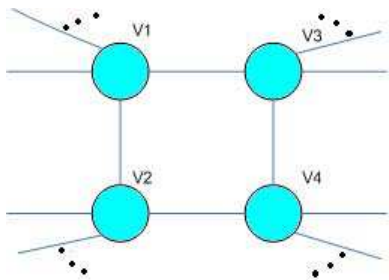}
\caption{Exception case when $B \ni v_1,v_2,v_3,v_4$ is not contractible}
\label{fig:exceptionFigure}
\end{figure}

Once we found these boundary nodes, we can proceed and partition the network into two. We subsequently compute the Laplacian matrix, and check for rank deficiency as described in Section \ref{subsec:coverageHoleDetection}.

\subsection{\textbf{Complexity Analysis}}
\label{subsec:CovComlexity}
Algorithms in sensor networks depend on variety of factors such as inter-node communication, in-node processing, memory requirements and run time. These carry different costs depending on the context, and the communication  cost  is almost always dominant. We conduct a complexity analysis accounting for these pertinent points. The complexity also depends on the spatial arrangement of the nodes. For simplicity, we assume that the region of deployment is convex. We also focus on an average cost per node rather than the cost of the entire network. Most of the complexity of the detection/localization algorithm (and therefore the bottlenecks) depends on three factors 1) Evaluating the function $f$, 2) Finding $\max(f)$ and 3) Finding spectral radius of the Laplacian. Furthermore, since each iteration of the partitioning procedure sees half of the surviving nodes removed, the average cost per node primarily depends on the first iteration.

\subsubsection{\textbf{Communications}}
\label{subsubsec:Communicaitons}
For evaluating $f$, each node discovers every other node at some point and broadcasts the information to its neighbors. Each node broadcasts the discovery precisely once for every other node. As a result, the complexity per node for evaluating $f$ is $o(n)$, where $n$ is the number of nodes.  Evaluating the complexity for determining $\max(f)$ is rather peculiar since the behavior of the node depends on value of $f$ at that node. Recall from Section \ref{sec:DiamterNodes} titled ``Diameter nodes'',  that if a node with a function value higher than any previously recorded value is encountered, this information is broadcast. The nodes with the highest value of $f$ for example, never broadcast anything during this part of the algorithm. In order to evaluate the complexity in this case, we consider a simple case where nodes are deployed in a circular region. The radius of this circle will be $\rho \propto \sqrt{n}$. In this case, the nodes which lie on the circle of radius  $r$, (see Figure \ref{fig:diagramForComplexity}) will broadcast the discovery of exactly $\rho -r$ nodes which have a higher $f$ value than previously recorded.
\begin{figure}
\centering
\includegraphics[width=0.4\textwidth]{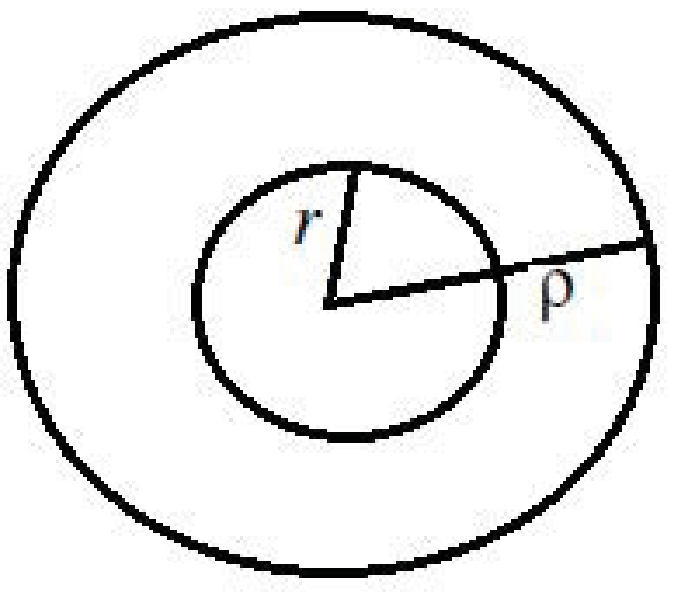}
\caption{Figure illustrating the simple case for assessing complexity}
\label{fig:diagramForComplexity}
\end{figure}

 The number of broadcasts over the entire network would end up as
$$ b = 2\pi\sum_{i=1}^{\rho}{i\left(\rho-i\right)}  = \frac{\rho^2\left(\rho+1\right)}{2}-\frac{\rho\left(\rho+1\right)\left(2\rho+1\right)}{6} = \frac{\rho\left(\rho-1\right)\left(\rho+1\right)}{6} \propto \rho^3 \propto n^{3/2}.$$
The average complexity per node for evaluating $\max(f)$ is therefore $o\left(\sqrt{n}\right)$. Figure \ref{fig:complexityDiameter} shows a $\log-\log$ relation between the number of nodes $n$ and the number of memory words broadcast for finding the diameter nodes in the first partition. For each value of $n$, we averaged over 5 networks. A linear regression (line in blue) shows a slope of $0.9\approx 1$ confirming the dominant effect of evaluating $f$ at $o(n)$ cost.  The complexity for finding the spectral radius of $L_1$  will be proportional to the mean \emph{degree} of the nodes (as the number of values broadcast in the power iteration method will be proportional to the number of neighbors) and depend logarithmically on the ratio $\alpha_1/\alpha_2$ where $\alpha_1$ and $\alpha_2$ are the first and second largest Eigen values. The difficulty of apriorily estimating these Eigen values  for a random matrix, will complicate this ratio as an explicit function of $n$. We therefore provide some numerical results shown in Figure \ref{fig:complexityLocalizeHoles}. This figure compares the total number of memory words broadcast for detecting a hole (evaluating spectral radius of $L_1$ and $\alpha I - L_1$) in our algorithm, with those required for localizing a  hole by an $l_1$ norm minimization as presented in \cite{DistLocHoles}.

\begin{figure}
\centering
\subfigure[Complexity for finding diameter nodes]{
    \label{fig:complexityDiameter} \includegraphics[width=0.3\textwidth]{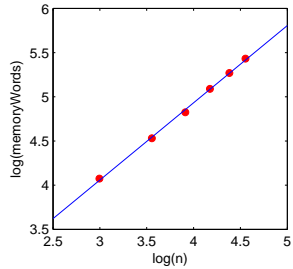}
    }
    \hspace{2 cm}
\subfigure[Complexity for localizing holes]{
    \label{fig:complexityLocalizeHoles} \includegraphics[width=0.3\textwidth]{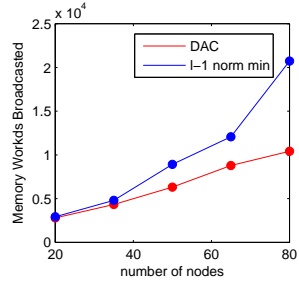}
    }
\caption{Complexity analysis for (a)finding diameter nodes (b)localizing holes. (b) compares the 'divide and conquer' method with $l_1$ optimization in \cite{DistLocHoles}}
\label{fig:ComplexityAnalysis}
\end{figure}

\subsubsection{\textbf{Memory}}
The only bottle neck in our algorithm for memory requirement lies in estimating the diameter points of the partition. At the $i^{th}$ time step during this process, a node discovers all the other nodes in the partition which are at a  $i-$ hop distance. The node keeps this information for 2 time steps, and deletes it. The iso-distance paths on the network from any node are on the order of $o\left(\sqrt{n}\right)$, translating this into the memory requirement for the algorithm.

\subsubsection{\textbf{Run Time}}
First, note that the number of partitions required to converge on to a hole is related to the number of nodes as $o\left(\log(n)\right)$. The time required for finding both the diameter nodes and the boundary nodes is directly proportional to the diameter of the network, i.e., $o\left(\sqrt{n}\right)$. The number of iterations required for the power iteration  method (for computing spectral radius of the Laplander) to converge, similar to its communication cost, is of the order $o\left(\log\left(\alpha_1/\alpha_2\right)\right)$. In each iteration of the power method, finding the sum (for normalizing) requires $\Theta\left(n\log(n)\right)$ time \cite{RandomWalks}\cite{RGA}.

\subsection{\textbf{Simulation Examples}}
\label{sebsec:CovSimulations}
Figure \ref{fig:PartitionProcess} shows the algorithm on a random network with 50 nodes. Figure \ref{fig:survivor1} shows the communication graph superimposed on the coverage area. In the first partition, the boundary nodes are indicated by the red circles and the diameter nodes are indicated in black. The boundary nodes dictate where the partition occurs and as shown in Figure \ref{fig:survivor2}, all the nodes in the partition which do not enclose a hole are no longer considered. An important point is that as the algorithm progresses, additional nodes are put to rest saving valuable power. IN the end, only the cycle closest to the coverage hole survives, providing a good indication of where the failure took place.

\begin{figure}
\centering
\subfigure[Communication Graph superimposed on the coverage area]{
    \label{fig:survivor1}\includegraphics[width=0.3\textwidth]{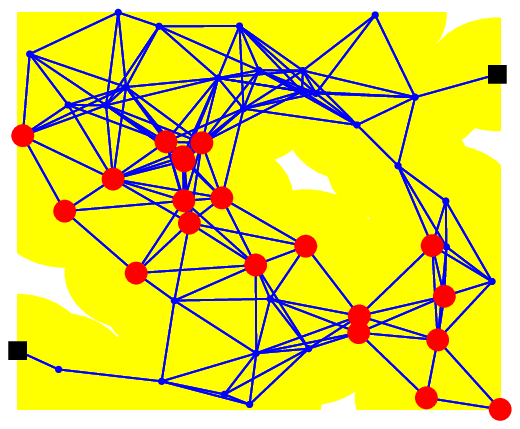}
    }
\subfigure[after partition 1]{
    \label{fig:survivor2}\includegraphics[width=0.3\textwidth]{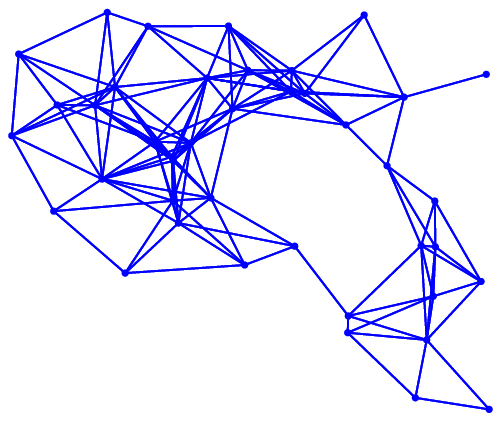}
    }
\subfigure[after partition 2]{
    \label{fig:survivor3}\includegraphics[width=0.3\textwidth]{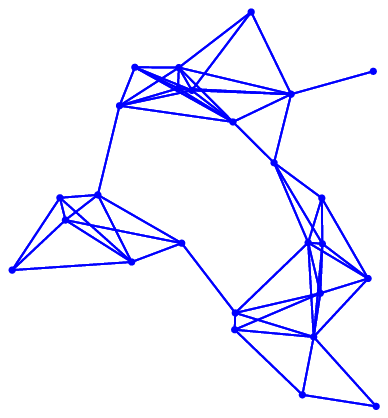}
    }
\subfigure[after partition 3]{
    \label{fig:survivor4}\includegraphics[width=0.3\textwidth]{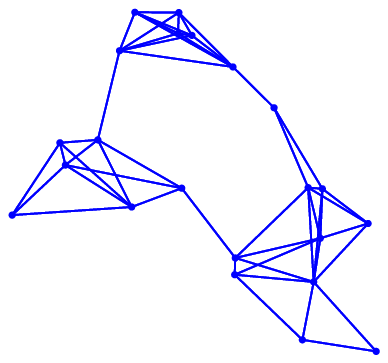}
    }
\subfigure[after partition 4]{
    \label{fig:survivor5}\includegraphics[width=0.3\textwidth]{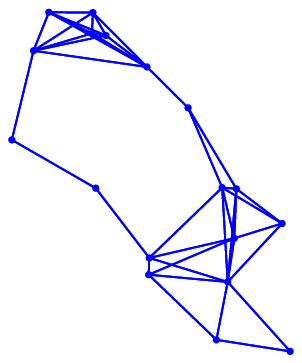}
    }
\subfigure[after partition 5]{
    \label{fig:survivor6}\includegraphics[width=0.3\textwidth]{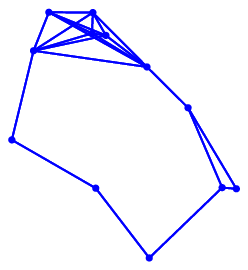}
    }
\subfigure[after partition 6]{
    \label{fig:survivor7}\includegraphics[width=0.3\textwidth]{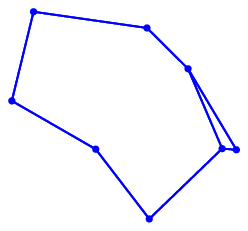}
    }
\subfigure[after partition 7]{
    \label{fig:survivor8}\includegraphics[width=0.3\textwidth]{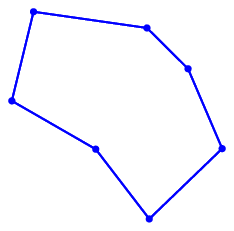}
    }
\caption{Figure showing the sequence of surviving supgraph with each partition}
\label{fig:PartitionProcess}
\end{figure}

\section{\textbf{Worm Hole Problem}}
\label{sec:WormHoleProblem}
A worm hole attack is launched by two colluding external attackers which do not authenticate themselves as legitimate nodes to the network. When starting a wormhole attack, one attacker overhears packets at one point in the network, tunnels these packets through the wormhole link (external to the network) to another point in the network. This generates a false scenario that the original sender is in the neighborhood of the remote location. An example of a worm-hole attack is shown in Figure \ref{fig:wormHoleDemonstration}. In this Section, our aim is to  first show the methodology of detecting, if such an attack is taking place, and if so, to locate the attack positions. By way of a simple observation, we show that the algorithm  to find a coverage hole, may be extended to address this problem.

\subsection{\textbf{Worm Hole Detection}}
\label{subsec:WormHoleDetection}
Because a worm hole links  geographically separated positions in the network, it essentially creates a cycle in the network which cannot be a boundary of a 2-simplex. It thus creates a non-zero homology component. Figure \ref{fig:WormHole3D} shows a network with a worm hole link, and the resulting deformation of the network structure which yields a cycle. We have already seen in Section \ref{sec:CoverageProblem} how to localize this cycle. It is hence clear that a presence of a worm-hole in the network, would be followed by a localization of the shortest cycle it creates.

\begin{figure}[!h]
\centering
\subfigure[network grid with links caused because of a worm-hole]{
    \label{fig:wormHoleGrid} \includegraphics[width=0.5\textwidth]{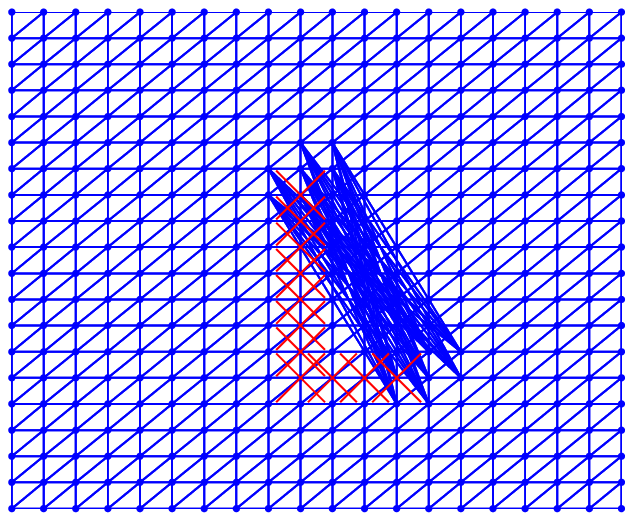}
    }
    \hspace{2 cm}
\subfigure[the same grid shown in 3d to respect distant properties measured as hop distances]{
    \label{fig:WormHole3D} \includegraphics[width=0.5\textwidth]{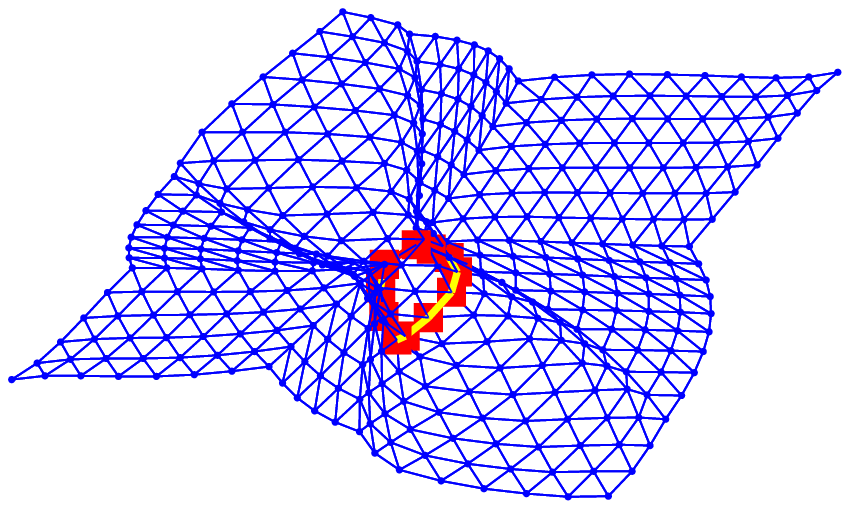}
    }
\caption{Deformation in network structure because of a worm hole. The cycle created is shown in red}
\label{fig:WormHoleFormation}
\end{figure}

The problem now reduces to identifying this cycle as a coverage hole or a wormhole. To this end, we formulate a simple algorithm shown in Table \ref{Tab:DetectingWormHoles}. The central idea of our approach is based on the observation that a cycle surrounding a coverage hole lies on the surface in which the network is deployed, while a cycle created by a worm hole will not lie in this surface. Removing a cycle which surrounds a coverage hole will therefore divide the network into two components, while removing a cycle created by a worm-hole does not. Figure \ref{fig:CycleGrowing} demonstrates this case. The network grid in Figure \ref{fig:cyleSurroundingCovHole} shows a coverage hole and the shortest cycle surrounding it. This cycle is grown homologously, i.e., without creating any more loops or destroying any, in the network as shown in Figure \ref{fig:cylceCovGrown}. When the nodes on this cycle along with their neighbors are removed from the network, the resulting network consists of two connected components as shown in Figure \ref{fig:cycleCovRemoved}. Figures \ref{fig:cycleWormHole}, \ref{fig:cycleWormGrown} and \ref{fig:cycleWormRemoved} show the same processes for a cycle created by a worm hole. As seen in \ref{fig:cycleWormRemoved}, the resulting network is still a single connected component. The two steps 1) Growing the cycle and 2) removing the nodes along with their neighbors, are explained in detail in the following sections.

\begin{table}[!h]
\label{Tab:DetectingWormHoles}
\centering
\begin{tabular}{l}
\hline
Grow the current cycle to get a longer homologous cycle in the network.\\
Remove the longer cycle along with its neighbors.\\
\hspace{0.25in} \emph{if} The above process creates an isolated component, \\
\hspace{0.50in} The cycle corresponds to a coverage hole.\\\
\hspace{0.25in} \emph{else}\\
\hspace{0.50in} The cycle corresponds to a worm hole.\\
\hline
\end{tabular}
\caption{\textbf{Algorithm for Detecting Worm Holes}}
\end{table}

\begin{figure}
\centering
\subfigure[Cycle surrounding a coverage hole]{
    \label{fig:cyleSurroundingCovHole}\includegraphics[width=0.2\textwidth]{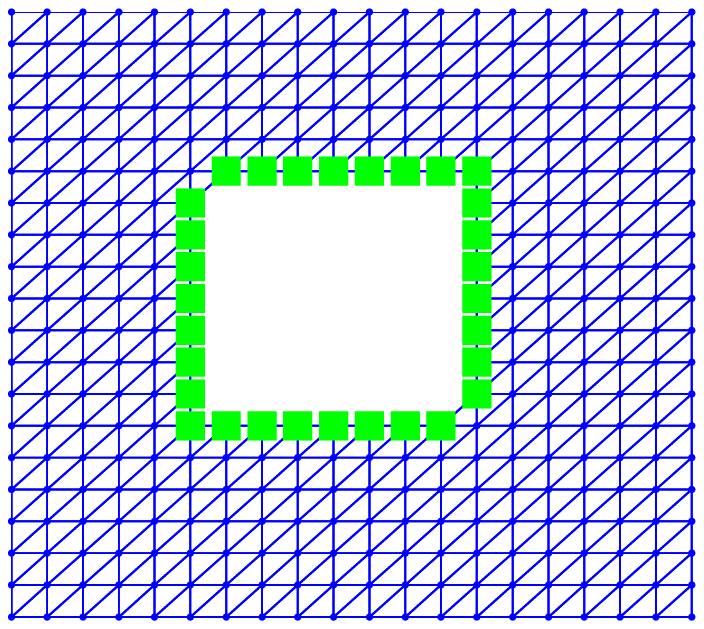}
    }
    \hspace{1 cm}
\subfigure[A cycle grown in the network]{
    \label{fig:cylceCovGrown}\includegraphics[width=0.2\textwidth]{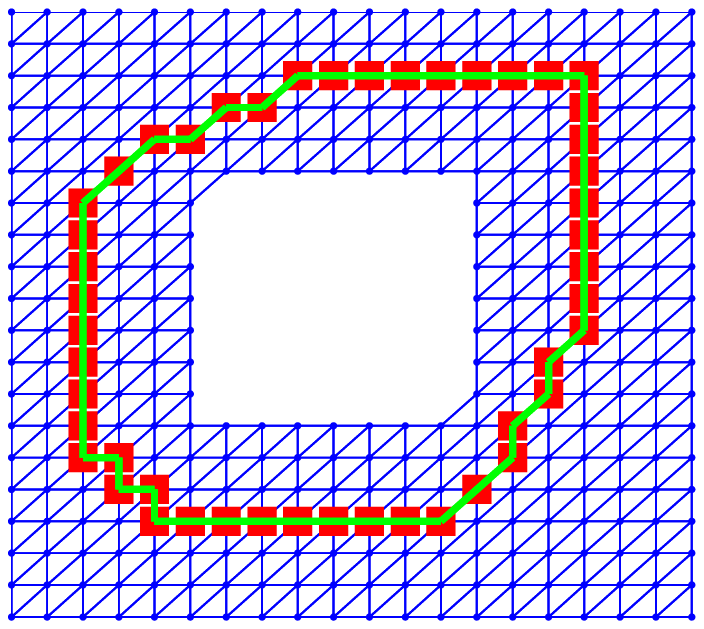}
    }
    \hspace{1 cm}
\subfigure[Grown cycle removed]{
    \label{fig:cycleCovRemoved}\includegraphics[width=0.2\textwidth]{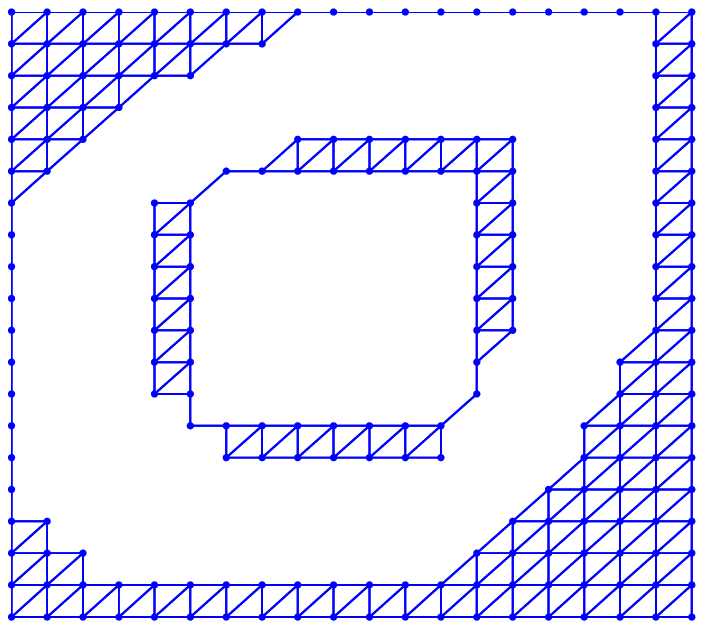}
    }
    \par
\subfigure[Cycle created by the wormhole]{
    \label{fig:cycleWormHole}\includegraphics[width=0.3\textwidth]{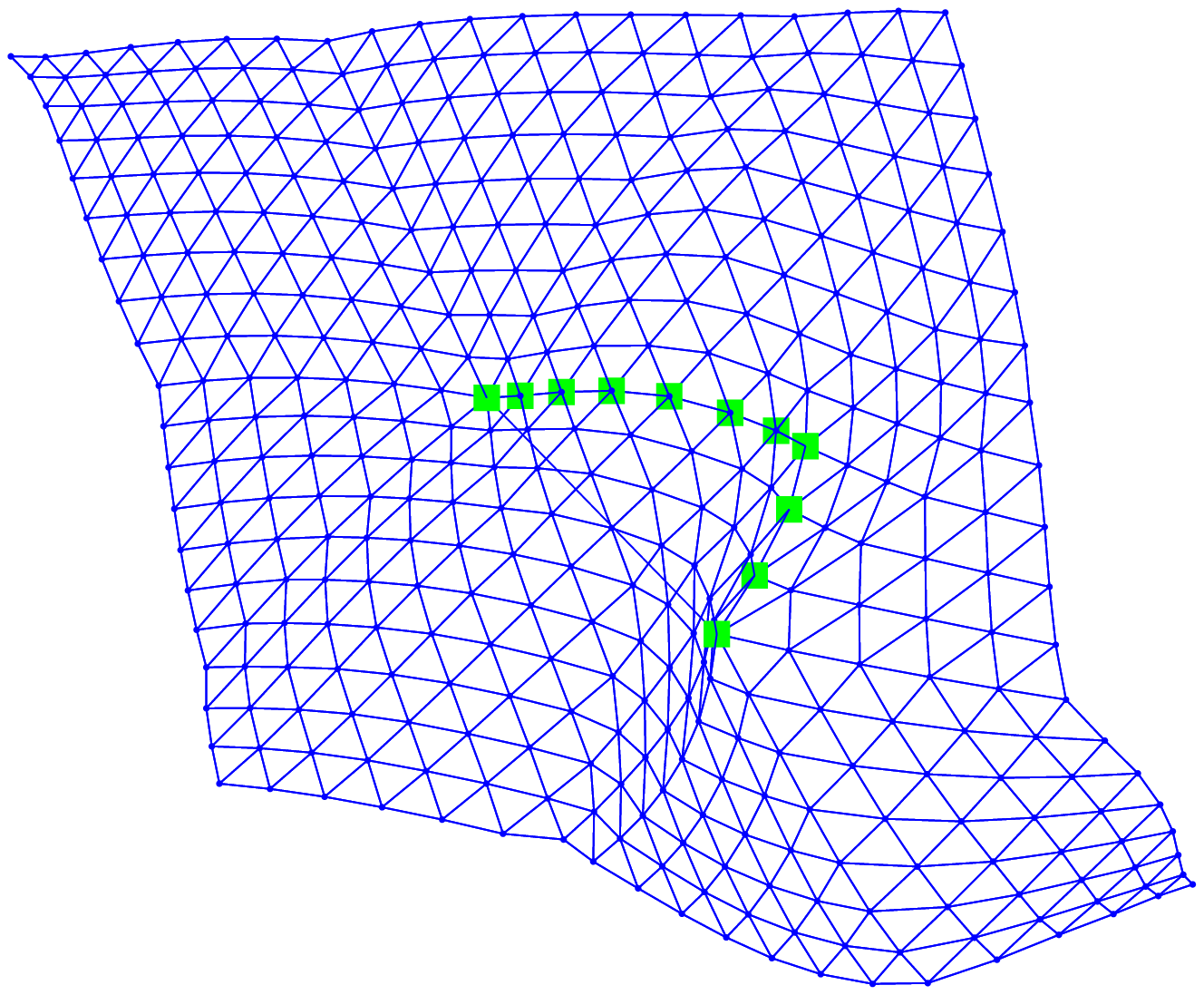}
    }
\subfigure[The cycle grown in the network]{
    \label{fig:cycleWormGrown}\includegraphics[width=0.3\textwidth]{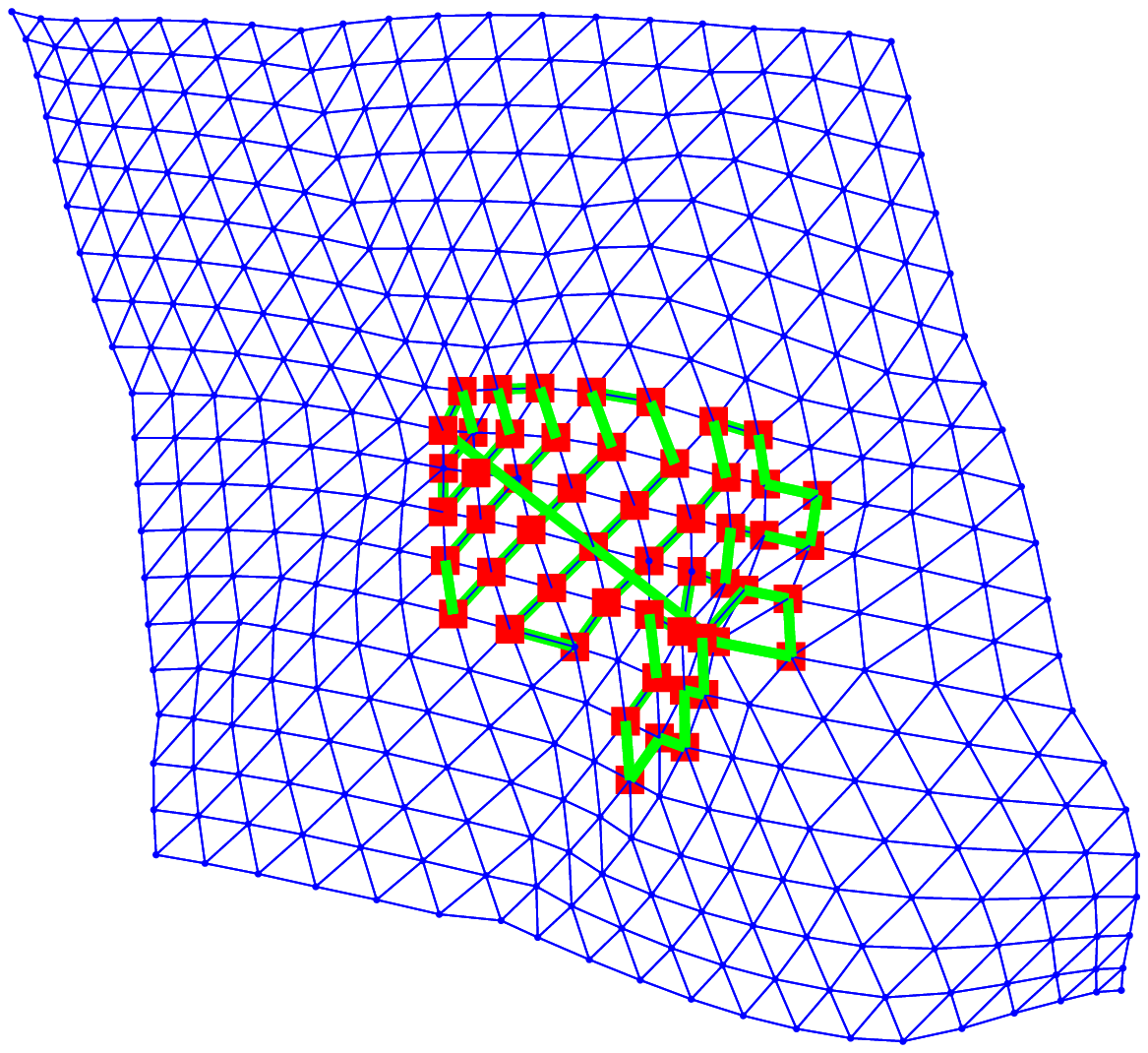}
    }
\subfigure[The grown cycle removed]{
    \label{fig:cycleWormRemoved}\includegraphics[width=0.3\textwidth]{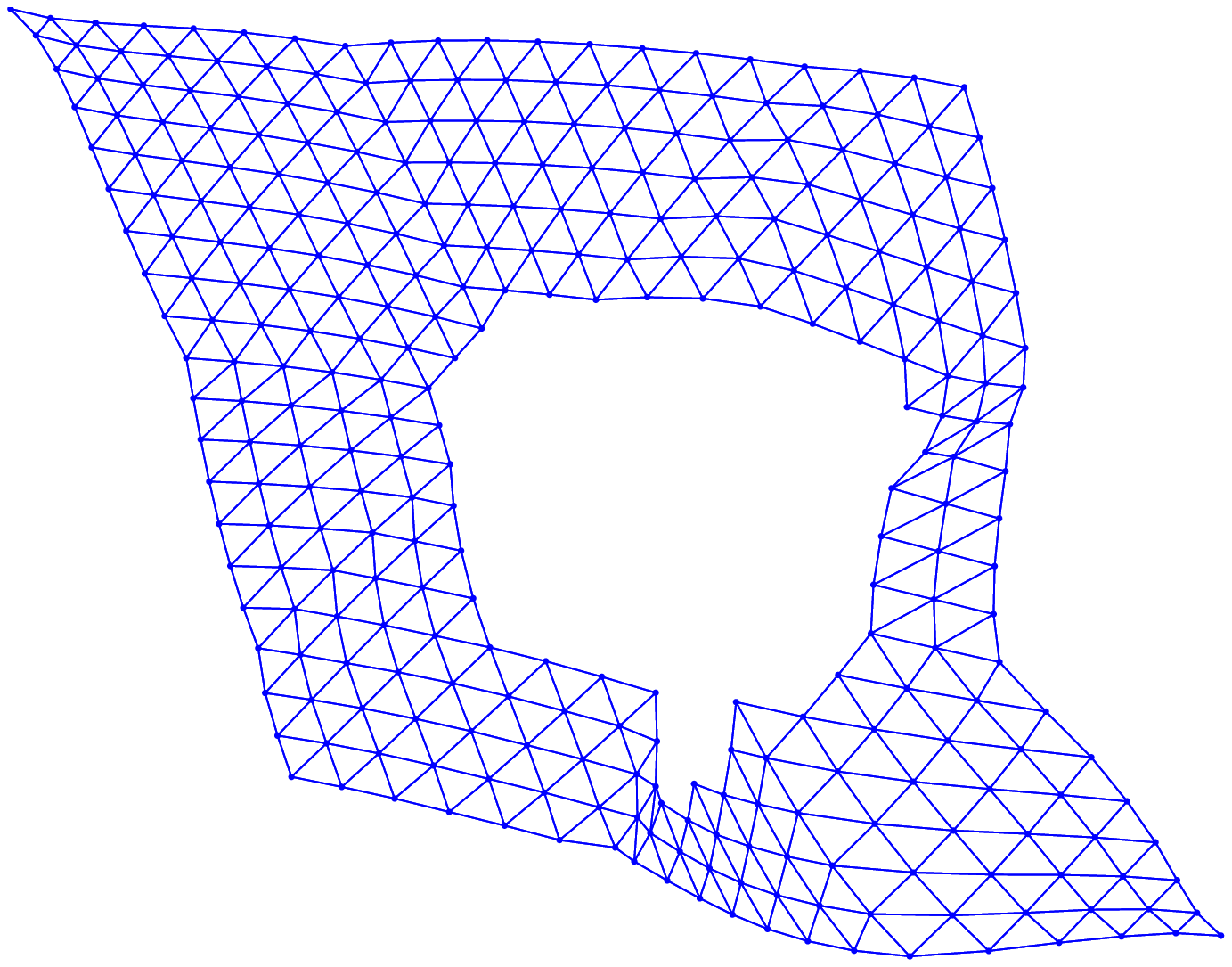}
    }
\caption{The structural difference between the cycles created by a coverage hole and by a worm hole}
\label{fig:CycleGrowing}
\end{figure}

\subsubsection{\textbf{Growing the Cycles}}
\label{subsubsec:GrowingCyles}
The algorithm described in Section \ref{sec:CoverageProblem} yields the shortest (in sense of hop-distance) non-contractible cycle which, in case of a coverage hole, leads to the boundary. Such a boundary will not serve our purpose, as removing a boundary will not partition the network into two. We hence have to grow this cycle "into" the network. The important properties which we have to abide by in the course of this cycle growth are
\begin{itemize}
\label{item:growingProperties}
\item We should not break a cycle at any time
\item We should not introduce any additional loops into the cycle.
\end{itemize}
A cycle which was originally surrounding a coverage hole will not do so after  it is broken. If we further introduce loops into the cycle during the growing procedure, a cycle due to a worm hole will now be similar to that surrounding a coverage hole. The above two properties are precisely captured by the idea of homologous chains. Two chains which belong to the same equivalence class in the homology space, are said to be homologous. Recall from Section \ref{subsec:homologicalAlgebraBackground}, the definition of homology groups as $H_k(C_{*}) = ker(\partial_k)/Img(\partial_{k+1})$. If $c_1, c_2 \in \ker\left(\partial_1\right)$ belong to the same equivalent class, then $c_1 - c_2 \in Img(\partial_2)$, i.e., their difference can be written as sum of the boundaries of 2-simplices (triangles). To that end, we "homologously" grow the cycle by applying two elementary steps, both of which add a boundary of a 2-simplex to the existing chain.
\\ \\
\textbf{Elementary Step 1}
If two adjacent nodes $v_1, v_2$ in the chain share a common neighbor $v_3$, we then remove the edge $(v_1,v_2)$ from the chain, and add the edges $(v_1,v_3)$ and $(v_3,v_2)$. This step is shown in Figure \ref{fig:elementaryStep1}.
\begin{figure}[!h]
\centering
\subfigure{
    \label{fig:elementaryStep1a}\includegraphics[width=0.17\textwidth]{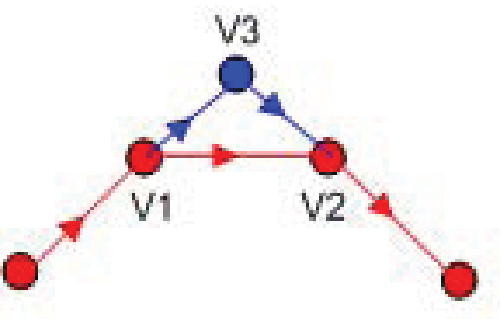}
    }
    \quad \quad \quad \quad
\subfigure{
    \label{fig:elementaryStep1b}\includegraphics[width=0.17\textwidth]{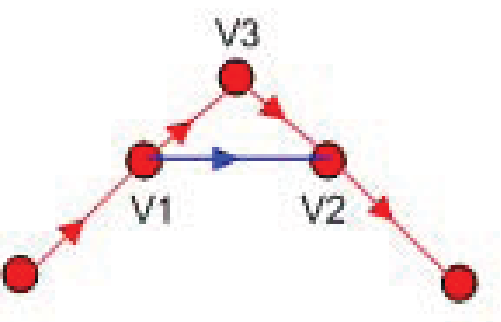}
    }
    \quad \quad \quad \quad
\subfigure{
    \label{fig:elementaryStep1c}\includegraphics[width=0.13\textwidth]{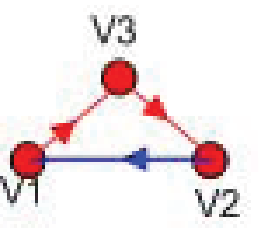}
    }
\caption{ (a)before, (b)after Elementary Step 1, and (c)the difference of the chains before and after the step. Note that (c), which is $c_1-c_2$, is a boundary of a 2-simplex. The edges shown in red are part of the chain.}
\label{fig:elementaryStep1}
\end{figure}
\\\textbf{Elementary Step 2}
If a node $v_1$ on the chain has two neighbors $v_2,v_3$ which are also on this chain, and $v_2, v_3$ are neighbors, we then remove the edges $(v_2,v_1)$, $(v_1,v_3)$ and we add the edge $(v_2,v_3)$. This step is shown in Figure \ref{fig:elementaryStep2}

\begin{figure}[!h]
\centering
\subfigure{
    \label{fig:elementaryStep2a}\includegraphics[width=0.2\textwidth]{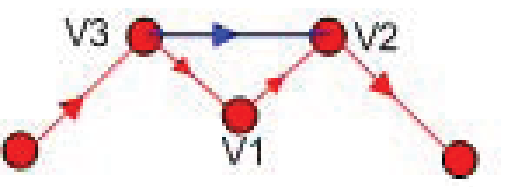}
    }
    \quad \quad \quad \quad
\subfigure{
    \label{fig:elementaryStep2b}\includegraphics[width=0.2\textwidth]{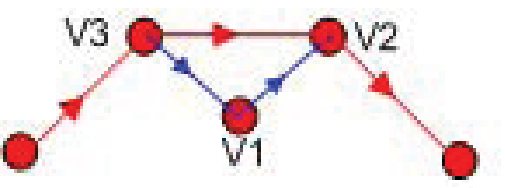}
    }
    \quad \quad \quad \quad
\subfigure{
    \label{fig:elementaryStep2c}\includegraphics[width=0.15\textwidth]{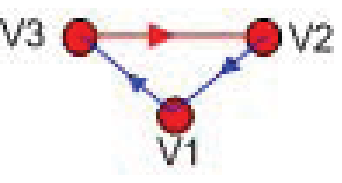}
    }
\caption{(a)before, (b)after Elementary Step 2, and (c)the difference of the chains before and after the step. Note that (c), which is $c_1-c_2$, is a boundary of a 2-simplex. The edges shown in red are part of the chain.}
\label{fig:elementaryStep2}
\end{figure}

\begin{figure}[!h]
\centering
\includegraphics[width=0.2\textwidth]{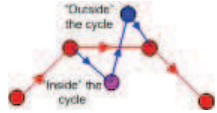}
\caption{An edge (not in the cycle) "crosses" one on the cycle. In this case, removing just the nodes on the cycle would not separate the network into two components.}
\label{fig:edgeCrossing}
\end{figure}

\begin{figure}[!h]
\centering
\subfigure[$v_1,v_2$ not in vicinity of $X$ and $Y$. A shortest path can be found in the network surrounding the nodes removed.]{
    \label{fig:EdgeRemovedNonVicinity}\includegraphics[width=0.6\textwidth]{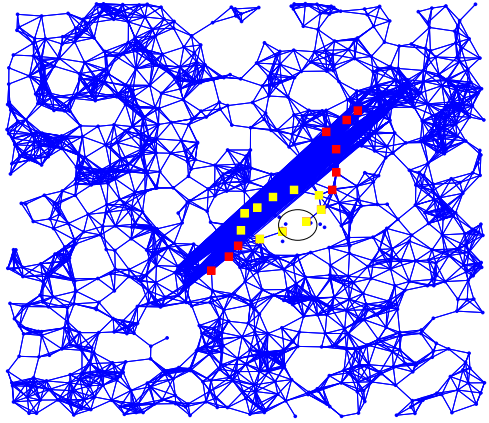}
    }
    \hspace{0.5 cm}
\subfigure[$v_1,v_2$ in vicinity of $X$ and $Y$. Alternative shortest path includes all the nodes in the cycle]{
    \label{fig:EdgeRemovalAtVicinity}\includegraphics[width=0.6\textwidth]{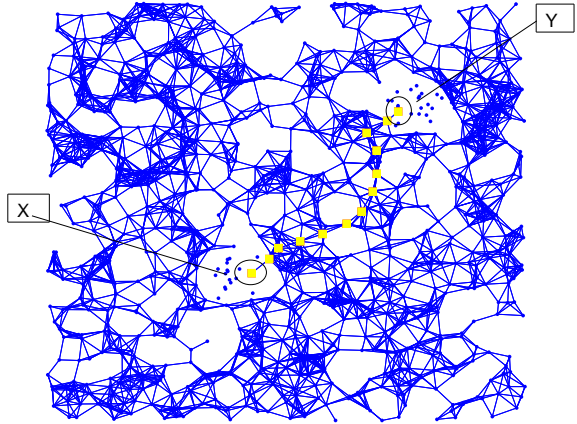}
    }
\caption{Worm Hole Localization}
\label{fig:WormHoleLocalization}
\end{figure}

\begin{figure}[!h]
\centering
\includegraphics[width=0.6\textwidth]{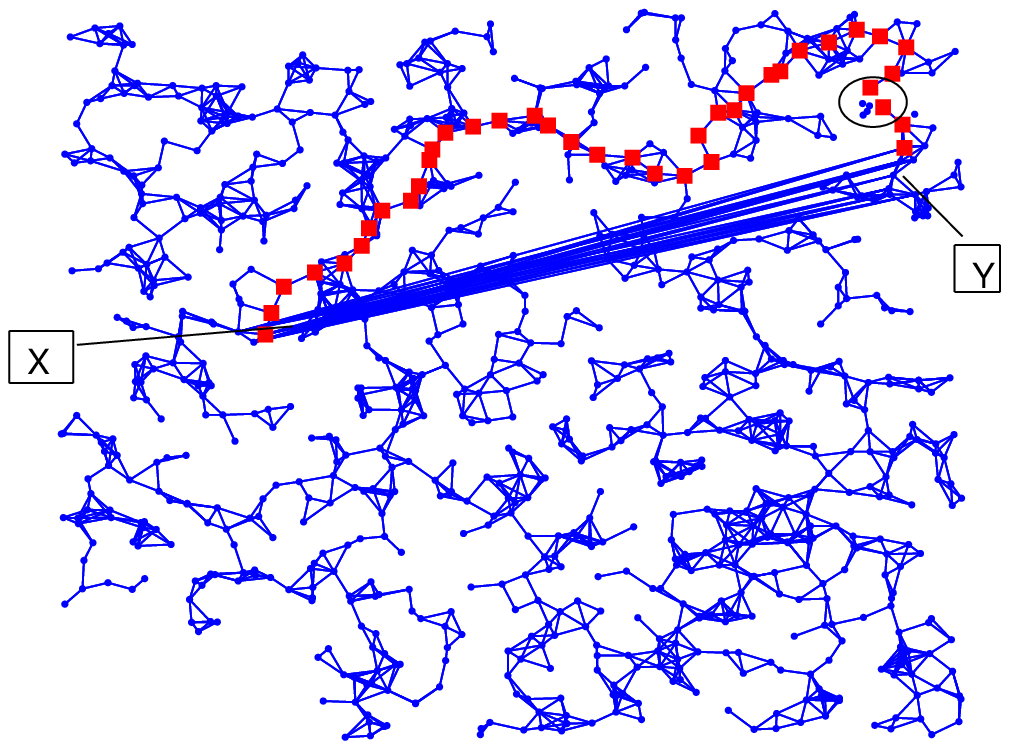}
\caption{A sparse network where the algorithm fails. The removal of nodes around $v_1$ and $v_2$ creates an isolated component in the network. }
\label{fig:sparseWormHole}
\end{figure}

\subsubsection{\textbf{Removing the cycle}}
\label{subsubsec:RemoveCycle}
We saw in Figure \ref{fig:CycleGrowing} that removing the nodes on a cycle surrounding a coverage hole, \emph{along with their neighbors}, yields two disconnected components. The reason for the insufficiency of removing just the nodes on the cycle, is that there might be two adjacent nodes in the graph whose edge "crosses" an edge on this cycle as shown in Figure \ref{fig:edgeCrossing}. However, it can be shown that removing the nodes on a cycle along with their neighbors, will result  in the  required separation \cite{udgreference}.
\par Note that until this point, we have made no assumption about the density of the nodes in the network. We can successfully identify whether the shortest non-contractible cycle identified in Section \ref{sec:CoverageProblem} corresponds to a coverage hole or a worm hole, i.e., we have detected if there is a worm hole in the network. We have thus far restricted the location of a worm hole attack to a relatively small subgraph of the network (the shortest cycle). In the following section, we present an effective approach to precisely locate the worm hole in question.

\subsection{\textbf{Worm Hole Localization}}
\label{subsec:WormHoleLocalization}
In order to precisely locate a worm hole, we first closely examine its impact. Denote the colluding nodes in the attack as $X$ and $Y$ (Figure \ref{fig:wormHoleDemonstration}). As a result of this attack, a node in the vicinity of $X$ considers all the nodes in the vicinity of $Y$ as its neighbors. This also results in the formation of a cycle. Observe that a simple way to undo the effect of a worm hole is to remove all the nodes in the vicinity of $X$ and $Y$. Note that the algorithm in Section \ref{sec:CoverageProblem} finds the shortest cycle, implying that there will exactly be two nodes on this cycle which are in the vicinity of $X$ or $Y$. In this light, we propose a simple algorithm given in Table \ref{Tab:WormHoleLocalizationAlgo} to localize the worm hole.

\begin{table}[!h]
\label{Tab:WormHoleLocalizationAlgo}
\centering
\begin{tabular}{l}
\hline
for each adjacent pair $v_1,v_2$ in the cycle \\
\hspace{0.25in} Remove the edge $(v_1,v_2)$ and all neighbors of $v_1$ and $v_2$ \\
\hspace{0.50in} except those on the cycle.\\
\hspace{0.25in} Find the shortest path between $v_1$ and $v_2$. \\
\hspace{0.25in} \emph{if} This shortest path coincides with the nodes on the cycle\\
\hspace{0.50in} $v_1$ and $v_2$ are in the vicinity of $X$ and $Y$.\\
\hline
\end{tabular}
\caption{\textbf{Algorithm for Localizing Worm Holes}}
\end{table}

If $v_1,v_2$ were indeed in the vicinity of $X$ and $Y$, then removing all their neighbors would remove all the spurious links caused by the worm hole. In this case, the shortest path between $v_1$ and $v_2$ would be the rest of the cycle. If on the other hand, they were not in the vicinity of $X$ or $Y$, then they would find an alternative path in the network which surrounds the deleted nodes. The result of this algorithm is shown in Figure \ref{fig:WormHoleLocalization}.
\par We note that this algorithm assumes a minimal node density to properly perform. It should however be noted that the algorithm is most effective when the network is sufficiently dense. When the network is sparse, the removal of neighbors of $v_1$ and $v_2$ may eliminate all possible paths between them, at the exception of those going through the links created by the worm hole. For example, in Figure \ref{fig:sparseWormHole}, since the network is very sparse, the removal of neighbors of $v_1$ and $v_2$ creates an isolated component, which is only linked by the worm hole. Any path will therefore have to go through one of the links created by the worm hole.

\section{\textbf{Conclusion}}
\label{sec:Conclusion}
In this work, we addressed using two specific problems; 1) Coverage Hole Localization and 2) An extended application to Worm Hole attack Localization. We have shown that topological analysis of a network provides us with substantial and ample information to assess its health, and requires minimal prior information. To that end, we have proposed an  Algebraic Topological approach as an elegant and efficient avenue for extracting useful information.  The formulation into an algebraic domain enables us to utilize extensive existing tools to effectively address these problems.  We have also, by way of the computational efficiency of our proposed approach, addressed a very crucial problem in sensor networks, namely that of prolonging the battery life of the nodes. \\

\bibliography{holeWormJournal}
\bibliographystyle{plain}

\end{document}